\documentclass[12pt]{article}
\usepackage{amsthm,amsmath,amsfonts,amssymb,cases}
\usepackage{comment}
\usepackage{enumerate}
 \usepackage{color}
\usepackage{tikz}
\usetikzlibrary{decorations.pathreplacing}
\usepackage{pgfplots}

\numberwithin{equation}{section}

\newcommand{\R}{{\mathord{\mathbb R}}}
\newcommand{\Z}{{\mathord{\mathbb Z}}}
\newcommand{\N}{{\mathord{\mathbb N}}}
\newcommand{\C}{{\mathord{\mathbb C}}}

\newcommand{\HH}{\mathcal{H}}

\newcommand{\FF}{\mathcal{F}}

\newcommand{\hh}{\mathfrak{h}}

\newcommand{\ran}{{\rm Ran}}

\def\one{{\sf 1}\mkern-5.0mu{\rm I}}

\def\inf{{\rm inf}\,}

\def\dist{{\rm dist}\,}

\newcommand{\inn}[1]{\left\langle {#1} \right\rangle }

\newtheorem{lemma}{Lemma}
\newtheorem{theorem}[lemma]{Theorem}
\newtheorem{remark}[lemma]{Remark}
\newtheorem{proposition}[lemma]{Proposition}

\newtheorem{definition}[lemma]{Definition}
\newtheorem{hyp}{Hypothesis}

\numberwithin{lemma}{section}

\begin{document}
\title{On one photon scattering in non-relativistic  qed}
\author{\vspace{5pt}  David G. Hasler\footnote{
E-mail: david.hasler@uni-jena.de} 
 \\
\vspace{-4pt} \small{Institute for  Mathematics,
University of Jena} \\ \small{Jena, DE}\\
}
\date{}
\maketitle

\begin{abstract}
We consider  scattering of  a single photon  by  an atom or a molecule in
the framework of non relativistic qed, and   we   express the scattering matrix  for one photon scattering
as a boundary value of the resolvent.
\end{abstract}

\section{Introduction}

Rayleigh scattering  describes the scattering of  light by bound  particles much
smaller than the wavelength of the radiation.   We study such a  process in the frame work of non-relativistic
qed.  Non-relativistic qed is a mathematical rigorous model  describing low energy aspects of  the quantized electromagnetic field  interacting with
non-relativistic quantum mechanical matter \cite{PauliFierz.1938}. Mathematical aspects of Rayleigh scattering    has been intensively studied for this and   related  models \cite{DerezinskiGerard.1999,FrohlichGriesemerSchlein.2001,FrohlichGriesemerSchlein.2002,DerezinskiGerard.2004,GriesemerZenk.2010,FaupinSigal.2014,DeroeckGriesemerKupiainen.2015}.  Existence of asymptotic creation  and annihilation operators as well as asymptotic completeness has been established.

 We define the scattering matrix   as the inner product  of   asymptotically incoming photon  states with    asymptotically outgoing photon
states.
Physically, the scattering  matrix   describes the emission and absorption of photons  by
an atom or molecule.   It can be used for example  to  determine   the absorption spectrum, observe
Bohr's frequency condition, and to detect properties of   resonances.
More specifically,
we  consider in this paper  the one photon scattering matrix and relate it to boundary values of the  resolvent.  The resulting
expression resembles  a special case  of a  LSZ  reduction formula  \cite{LehmannSymanzikZimmermann.1955,BjorkenDrell.1965}
obtained in   \cite[(III.13)]{BachFroehlichPizzo.2007}
after  performing  a formal  time integration.   The relation which we obtain, can  be used
to calculate expansions of the one photon scattering matrix using well established results about
 the ground state, resonance states, and  boundary values of the resolvent
\cite{BachFroehlichSigal.1998b,BachFroehlichSigal.1999,Sigal.2009,HaslerHerbstHuber.2008,GriesemerHasler.2009,HaslerHerbst.2011a,HaslerHerbst.2011b}.
We note that an analogous relation has been obtained for the spin boson model, where the
  leading order  expansion has been studied in connection with resonances \cite{BallesterosDeckertHaenle.2019,BallesterosDeckertHaenle.2020}.
For the hydrogen atom with a spinless electron  expansions of the scattering matrix have been studied in  \cite{BachFroehlichPizzo.2007} using the LSZ reduction  formalism
combined with iterated perturbation theory. In the present    paper we  provide  a different approach, which   gives   an explicit relation of the scattering matrix
to boundary values of the resolvent.
Thus it can be used to relate    resonance eigenvalues  to singularities of the scattering matrix.

To prove our result we make several assumptions, which have been established  in the literature in various  situations.
We assume the existence of a ground state, the  existence of wave operators, and regularity properties
of boundary values of the resolvent.
In Section \ref{modres} we introduce the model and state the main result.
The proofs are given in Section  \ref{sec:proofs}. In Section \ref{sec:veri}  we verify the  Hypotheses about boundary values of the resolvent for explicit models.

The idea of the proof of the main result  is as follows. First we define  the asymptotic creation and annihilation operators.
Then  we  express  the asymptotic creation operators using Cook's method as an integral over the  current, which is formulated     in Lemma  \ref{lem:cook}, cf. \cite{BachFroehlichPizzo.2007}. Using this identity
we relate matrix elements  of the current   with  one particle scattering states  and
 the  ground state    to the
T-matrix, which is stated     in Proposition \ref{eq:propTmat}.
 The $T$-matrix,  defined in \eqref{defofT}, is given as a  matrix element of boundary values of the resolvent.  The method can be viewed as a generalization   of  one particle scattering in quantum mechanics, cf.
 \cite{Simon.1971,ReedSimon.1979}   ,
 and is  analogous to the corresponding relation obtained in
 \cite{BallesterosDeckertHaenle.2020} for  the spin boson model. To show Proposition \ref{eq:propTmat} we use
an Abelian limit to integrate out the time evolution  by means of the following  identity for a self-adjoint operator $A$
\begin{align*} 
\int_0^\infty  e^{- i t (A-i\epsilon)} dt = \frac{-i}{A-i\epsilon}   ,
\end{align*}
which is understood as a strong
Riemann-integral  and $\epsilon > 0$ is a regularization. Then we use
a  pull-through formula, cf.  \eqref{eq:pull}, and assumptions about   boundary values of the resolvent to
take the limit $\epsilon \downarrow 0$. We then use once more  the integral representation of the asymptotic creation operators,
 Lemma   \ref{lem:cook},    their relation to the $T$-matrix,  Proposition \ref{eq:propTmat},
and the canonical commutation relations, which are  satisfied by    the  asymptotic creation and annihilation operators
\cite[Theorem 4]{FrohlichGriesemerSchlein.2001}.
Combining these results   with   assumed regularity properties about boundary values of the resolvent, we can
 interchange integrals and relate
the scattering matrix to the boundary value of the resolvent.

In Section \ref{sec:veri}, we verify the assumed regularity properties using analytic dilation \cite{BalslevCombes.1971,Simon.1973a}.
In particular,  we combine a result about exponential decay of analytic extensions of
the ground state  \cite{HaslerLejsek.2022} with estimates of analytic dilations of the  resolvent  \cite{BachFroehlichSigal.1999,HaslerHerbstHuber.2008}.

We note that in contrast to an analogous result for the simpler spin boson model with a mild infrared regularization
 \cite{BallesterosDeckertHaenle.2020},  we have  to deal with the electronic degrees of freedom,
which are described by an infinite dimensional space.  Furthermore,  we do not  impose any infrared regularization.
In  \cite{BallesterosDeckertHaenle.2020} a  multiscale analysis was central to the prove of  the result.
In the present paper we,  a priori, do not need such an analysis. 
 The emphasis of our proof is to isolate
abstract hypotheses which are needed for the main result.  Moreover, we show that these  hypotheses can be verified
 using  that the ground is dilation analytic and has spatial exponential decay  \cite{HaslerLejsek.2022}.
A result  which   is involved to show,  but is of interest of its own.
Nevertheless, we  believe that the hypotheses made in the present paper   could alternatively be  established  using more accessible
 methods such as   Mourre theory \cite{Mourre.1980}. Moreover, we want to mention that the ideas of the proof of   the present paper
  have been applied  to   the related  Nelson model \cite{Lejsek.2022}.

\section{Model and Statement of Result} \label{modres}

Let $\hh$ be a Hilbert space. We denote the Bosonic
Fock space over $\hh$ by $\mathcal{F}(\hh)$. We define $\hh^{\otimes n } := \bigotimes_{j=1}^n  \hh $ for $n \in \N$ and set $\hh^{\otimes 0 } = \C$.
Let $\mathfrak{S}_n$ denote the permutation group of $\{1,...,n\}$. For $\sigma \in \mathfrak{S}_n$ we define a linear  operator, also denoted  by $\sigma$,  on basis elements of $\hh^{\otimes n} $
by
$$
\sigma (\varphi_1 \otimes \varphi_2 \otimes \cdots \otimes \varphi_n ) =  \varphi_{\sigma(1)} \otimes \varphi_{\sigma(2)} \otimes \cdots \otimes \varphi_{\sigma(n)}  ,
$$
 extend  it first linearly and then by  taking  the  closure  to a bounded operator on $\hh^{\otimes n}$. Define $S_n = \frac{1}{n!} \sum_{\sigma \in \mathfrak{S}_n } \sigma$ and $S_0 =  \one$.
The Fock space is defined by
$$
\FF(\hh) := \bigoplus_{n=0}^\infty \FF_n(\hh) ,\quad \FF_n(\hh) := S_n  (\hh^{\otimes n} ) .
$$
For $h \in \hh$ the creation operator,  $a^*(h)$, is defined on elements $\eta \in \FF_n(\hh)$  by
$$
a^*(h) \eta    = (n+1)^{1/2}    S_{n+1}  ( h \otimes   \eta  ) ,
$$
and extends linearly to a densely defined closed linear operator in $\FF(\hh)$. Its adjoint, $a(h)$,
is called the annihilation operator.
It follows from the definition that the  operators $a^*(h)$ and $a(h)$ obey the usual canonical commutation relations (CCR)
$$
[a(g),a^*(h)] = \inn{ g , h } , \quad [ a^\#(g), a^\#(h) ] =  0 ,
$$
here $\inn{ g , h }$ denotes the inner product in $\hh$ and $a^\#(\cdot)$ stands for  either  $a^*(\cdot)$ or $a(\cdot)$.  We define the Bosonic field operator by
$$
\phi(h) =  \frac{1}{\sqrt{2}} \overline{  ( a(h) + a^*(h) )  }^{\rm cl}    ,
$$
where $\overline{( \cdot )}^{\rm cl}$ stands for the operator  closure.
In this paper we consider photons, i.e.,
relativistic transversally polarized Bosons, and  choose henceforth   $$\mathfrak{h} = L^2(\R \times \Z_2 ) \cong L^2(\R^3 ; \C^2) . $$
Let  $\omega(k) = |k|$.   The   operator of multiplication by  $| \cdot |$ shall  be denoted by  $\omega$ as well. The operator $H_{\rm f}$ in $\FF(\hh)$ is defined    by
$$
H_{\rm f}|_{\FF_n(\hh)}  =  \sum_{j=1}^n \left( \one^{\otimes (j-1)} \otimes  \omega  \otimes \one^{\otimes (n-j-1)} \right)
$$
and linearly extends to a self-adjoint operator in $\FF(\hh)$.
The Hilbert space of  $N \in \N$ non-relativistic particles  is assumed to be
\begin{align} \label{haattdeff}
\HH_{\rm at} := L_{ \#}^2((\R^{3} \times  \Z_{2s+1})^N)  ,
\end{align}
where $s \in \{ 0,1/2 \}$ denotes  the spin of the non-relativistic particles,   $L_\#^2$ stands for the  subspace of either symmetric or antisymmetric elements with respect to interchange of  particle coordinates, c.f., \eqref{eq:permfun}. We shall denote the spacial particle coordinates by   $x_j$, $j=1,...,N$.
 The Laplacian on $\R^{3N}$ is written as usual by the symbol $\Delta$.
The potential describing the forces acting on the non-relativistic particles is denoted by $V$. We assume  that   $V \in L_{\rm loc}^2(\R^{3N})$, that $V$ is   symmetric with respect to permutations of particle coordinates, and
that  $V_-  := \max(0,-V) $ satisfies the following  relative form bound.  For   all $\epsilon > 0$  there exists a constant $C_\epsilon$ such that
\begin{equation} \label{eq:potas}
V_- \leq \epsilon( - \Delta) + C_\epsilon .
\end{equation}
By $S_j$ we shall denote the spin operator of the $j$-th particle. If $s=0$, then  $S_j = 0$, and  if  $s=1/2$,   then
$$(S_j)_l = \one^{\otimes ( j-1)}  \otimes  \frac{1}{2} \sigma_l \otimes \one^{\otimes (N-j-1) } ,   \quad l=1,2,3,
$$
with  $\sigma_1,\sigma_2,\sigma_3$ denoting the Pauli matrices.
The Hilbert space of the total system is
$$
\mathcal{H} := \HH_{\rm at} \otimes \FF(\hh)     .
$$
 To describe the interaction we
 introduce  the following  weighted $L^2$-space. For  $h \in L^2(\R^3;\C^2) \cong L^2(\R^3 \times \Z_2)$ we define
$$
\| h \|_\omega := \left( \sum_{\lambda=1,2} \int |h(k,\lambda)|^2 \left(1+ |k|^{-1} \right) dk \right)^{1/2}
$$
and
$$
\hh_\omega := L_\omega^2(\R^3 \times \Z_2)   := \{ h \in  L^2(\R^3 \times \Z_2)  : \| h \|_\omega  < \infty \} .
$$
The interaction between non-relativistic particles and the quantized field is described by the quantized vector
potential in Coulomb gauge and the quantized magnetic field, which are defined  by
$$
A_l(x) = \phi(G_{x,l})  \quad  \text{ and } \quad B_l(x) = \phi(J_{x,l})  , \quad x \in \R^3, \quad l=1,2,3,
$$
respectively, where
\begin{align*}
G_{x,l}(k,\lambda)  & = \frac{1}{(2\pi)^{3/2}} \frac{\kappa(k)}{\sqrt{|k|}} \varepsilon_l(k,\lambda) e^{-ik \cdot x} , \\  J_{x,l}(k,\lambda)  & =\frac{ 1}{(2\pi)^{3/2}}   \frac{\kappa(k)}{\sqrt{|k|}}  [ (-i  k)  \wedge \varepsilon(k,\lambda) ]_l e^{-ik \cdot x} ,
\end{align*}
$k  \in \R^3 \setminus \{ 0 \}$, and $\lambda \in  \Z_2$.
 Here,   $\varepsilon(k,1), \varepsilon(k,2) \in \R^3$  are the   polarization vectors.
We   assume that they  only depend on    $k/|k|$, for $k \in \R^3 \setminus \{ 0 \} $,   that  $(k/|k|,\varepsilon(k,1),\varepsilon(k,2))$ forms
an orthonormal basis  in $\R^3$, and that they are   measurable functions  of  their arguments.
 Moreover, we  assume that   $\kappa$ is  a measurable function from  $\R^3$  to $\C$ such that (or more precisely for its canonical extension $\kappa(k,\lambda) = \kappa(k)$)
  \begin{equation} \label{eq:assonfield}
   \omega^{-1/2} \kappa  \in L_\omega^2(\R^3 \times \Z_2) \quad  \text{ and if  } s = 1/2  , \, \text{ also }  \quad  \omega^{1/2} \kappa \in L_\omega^2(\R^3 \times \Z_2)  . \end{equation}
Physically, the function $\kappa$ is used as an ultraviolet cutoff and  incorporates  the coupling strength (Eq.   \eqref{eq:assonfield} holds  for example if
$\kappa(k) = e 1_{|k| < \Lambda}$ with $e \in \R$ the electric charge and $\Lambda \in (0,\infty)$ an ultraviolet cutoff).
The    Hamiltonian of the  total system is given by
\begin{equation} \label{eq:defofham}
H  = \sum_{j=1}^N  \{ (p_j + A(x_j))^2 +   \mu S_j \cdot B(x_j) \}  + V \otimes \one  + \one \otimes H_{\rm f} ,
\end{equation}
where $\mu \in \R$  (the  value for the standard model of non-relativistic qed is     $\mu = 2$) and $p_j = - i  \partial_{x_j}$.
 In \eqref{eq:defofham} and henceforth, we adopt the convention that for a vector of  operators $T =( T_1,T_2,T_3)$ we write $T^2$ for $ T \cdot T$. Furthermore, for notational compactness we shall suppress the tensor product in the notation of   operators, if it is clear from the context on which factor
an operator acts.
The definition  \eqref{eq:defofham} yields a priori a symmetric operator   on a  suitable dense subspace, $\mathcal{C}$,  of $\HH$ which we choose
to be equal  to the linear span of the set
$$\{ f \otimes  \eta  : f  \in \HH_{\rm at} \text{ with }  f  \text{ in }  C^\infty_c, \,  \eta \in S_n( C_c^\infty(\R^3 \times \Z_2)^{\otimes n } )  , n \in \N_0  \}.  $$
Using  assumptions \eqref{eq:potas}  and \eqref{eq:assonfield} it follows from Lemma    \ref{FGS.Lemma18}, in the appendix, that the operator is bounded from below.
This allows  to define
a self-adjoint operator  by means of  the Friedrichs extension, see for example \cite[Theorem X.23]{ReedSimon.1975}. 
Henceforth  we shall  denote by $H$  this self-adjoint extension. We assume that $H$ is essentially self-adjoint on $\mathcal{C}$ and that
\begin{equation}\label{domprob} \mathcal{D}(H) \subset \mathcal{D}((-\Delta +  H_{\rm f})^{1/2}) , \end{equation}
where the right hand side denotes the natural domain of  $(-\Delta +  H_{\rm f})^{1/2}$.
For a closed linear operator $T$ in a Hilbert space  we shall denote by $\sigma(T)$
is spectrum and by $\rho(T)$ its resolvent set.
Let
$$
E_{\rm gs} = \inf \sigma(H) .
$$

\begin{remark} \label{rem:def} {\rm  We note that if $V$ is infinitesimally bounded with respect to $-\Delta$, which is the case for Coulomb potentials, then
the assumptions on $V$ are satisfied. Clearly,  \eqref{eq:potas} holds. Furthermore,  $H$ is essentially self-adjoint on $\mathcal{C}$
and the domain of $H$ is equal to $\mathcal{D}(-\Delta + H_{\rm f})$,  the natural domain of $-\Delta  +H_{\rm f}$, \cite{Hiroshima.2002,HaslerHerbst.2008b}, and hence \eqref{domprob} is satisfied. 
}
\end{remark}

\vspace{0.2cm}

\noindent

Introducing  annihilation operators,  $a(k,\lambda)$, as  distributions defined  in Equation
\eqref{eq:defofpsi} in the appendix,  we can write  in the sense of forms on a dense subspace (e.g. finite linear combinations of finite tensor products of compactly supported smooth functions)  the quantized vector potential and magnetic field as
\begin{align}
A(x)  & = \frac{1}{(2\pi)^{3/2}}  \sum_{\lambda=1,2} \int    \frac{\varepsilon(k,\lambda) }{\sqrt{2 |k|}}  \{ \overline{\kappa(k)} e^{ i k  \cdot x } a(k,\lambda) +\kappa(k)
e^{-ik \cdot x} a^*(k,\lambda) \} dk  \label{hfdist1}  \\
B(x)  & =  \frac{1}{(2\pi)^{3/2}}  \sum_{\lambda=1,2} \int   \frac{i k \wedge \varepsilon(k,\lambda) }{\sqrt{2 |k|}}\{ \overline{\kappa(k)} e^{ i k  \cdot x } a(k,\lambda) - \kappa(k)
e^{-ik \cdot x} a^*(k,\lambda) \}  dk  ,
\end{align}
respectively,
and  the field energy as
\begin{align}\label{hfdist2}
H_{\rm f}  &=  \sum_{\lambda=1,2} \int  |k| a^*(k,\lambda) a(k,\lambda) dk .
\end{align}
Here,    $a^*(k,\lambda)$ denotes  the  formal adjoint of $a(k,\lambda)$
 and the integrals \eqref{hfdist1}--\eqref{hfdist2} are understood in the weak sense.

To state the main result, we shall introduce   several hypotheses.
It has been shown in the literature that these assumptions  hold  true in various situations.

\begin{hyp}\label{hyp:exgs} The number $E_{\rm gs}$ is an eigenvalue of $H$, i.e., the operator   $H$ has a square integrable   ground state.
\end{hyp}

Assuming Hypothesis \eqref{hyp:exgs} we will denote by   $\psi_{\rm gs}$
  a  normalized  eigenvector with eigenvalue $E_{\rm gs}$.

\begin{remark}{\rm For small coupling, i.e., $\|\omega^{-1/2} \kappa \|_\omega$ and  $\|\omega^{1/2} \kappa \|_\omega$  small, the
existence of ground states has been established in \cite{BachFroehlichSigal.1998b}.
For   large coupling existence  has been shown  in \cite{GriesemerLiebLoss.2001},
where   it is  assumed that $\kappa(k) = e 1_{|k| \leq \Lambda}$  with $e \in \R$ and $\Lambda \in [0,\infty)$
arbitrary.
}
\end{remark}

\vspace{0.2cm}

Let us now   introduce the so  called asymptotic creation and annihilation  operators.
They are defined for  $h \in \hh$ on  vectors $\psi \in \HH$     by
\begin{align}  \label{defofasy}
a_{\substack{ {\rm out} \\ {\rm in }}}^\#(h) \psi& = \lim_{t \to \pm \infty} e^{ i H t} e^{- i H_{\rm f} t } a^\#(h)   e^{ i H_{\rm f}  t} e^{- i H  t } \psi ,
\end{align}
provided the limit exists. One refers to  $a_{\rm in}^\#$  and $a_{\rm out}^\#$ as an incoming and  outgoing
asymptotic operator, and   \eqref{defofasy}  is called an  incoming and outgoing state, respectively.
 We note that as a consequence of the definitions it is straight forward to see that for  $h \in \hh$
\begin{equation}
\label{defofasy2}
e^{- i H_{\rm f} t } a^\#(h)   e^{ i H_{\rm f}  t} = a^\#( e^{ -  i \omega t } h )   .
\end{equation}
Furthermore, we shall assume the following hypothesis.

\begin{hyp}\label{thm:fgs01thm4} 
There exists an $E > E_{\rm gs}$, such that for all $g,h  \in \hh_\omega$
the following holds.
\begin{enumerate}[(i)]
\item  \label{(i)} For $\varphi \in \HH$ with  $\varphi = 1_{H \leq E} \varphi$  the  limits
$$
a_{\substack{ {\rm out} \\ {\rm in }}}^\#(h) \varphi = \lim_{t \to \pm \infty} e^{i H t}  e^{- i H_{\rm f} t} a^\#( h) e^{ i H_{\rm f} t}  e^{- i H t } \varphi
$$
exist.  Furthermore, there exists a constant $C > 0$ such that for all $f \in \hh_\omega$
\begin{equation} \label{formalex}
\| a_{\substack{ {\rm out} \\ {\rm in }}}^\#(f) 1_{H \leq E} \| \leq C \| f \|_\omega .
\end{equation}
\item \label{(ii)} The canonical commutation relations
$$
[a_{\rm  in}(g),a_{\rm  in}^*(h) ] = \inn{ g,h }  \quad \text{and} \quad [a_{\rm  in}^\#(g),a_{\rm in}^\#(h) ] =  0
$$
hold true, in form-sense, on $1_{H \leq E} \HH$. If $\psi \in 1_{H  =  E_{\rm gs}} \HH$, then
\begin{equation} \label{anngs}
a_{\rm in}(g)  \psi = 0 .
\end{equation}
\end{enumerate}
\end{hyp}

\begin{remark}{\rm
We note that Hypothesis \ref{thm:fgs01thm4} has been shown in
 \cite[Theorem 4]{FrohlichGriesemerSchlein.2001} in the spinless case.
There it was assumed that $\kappa$ has compact support and that the potential $V$ is given as a sum of one body
and two body potentials.
Specifically, we  note that  from the assertion  of  Theorem 4 in \cite{FrohlichGriesemerSchlein.2001} Eq.  \eqref{anngs}
follows  first for $g \in \hh_\omega$ with $m := \inf \{|k|: g(k) \neq 0 \} > 0$,
since $ a_{\rm in}(g)  {\rm Ran}  1_{H  =  E_{\rm gs}}  \subset {\rm Ran} 1_{H  \leq E_{\rm gs} - m } = \{ 0 \} $,
 where the inclusion holds by Part (iv) of Theorem  4    in   \cite{FrohlichGriesemerSchlein.2001}.
Then   \eqref{anngs}  extends to all $g \in \hh_\omega$ in view of  \cite[  Eq. (39) in Theorem 4]{FrohlichGriesemerSchlein.2001}.

}
\end{remark}

To formulate the main result, we define  the following commutators, which are  related
to the electric  current, see for example  \cite{BachFroehlichPizzo.2007}.
Calculating a  commutator we find for $\psi \in \mathcal{C} $ and $h \in \hh$, that
\begin{align}
    \sum_{j}\left[ ( p_j   +  A(x_j))^2  + \mu S_j  \cdot B(x_j)     ,  a^*( h ) \right]  \psi  &  =   \sum_{\lambda=1,2} \int_{\R^3}    h(k,\lambda) D_1(k,\lambda)    \psi  dk  \label{comm1}  ,
\end{align}
(using $k \cdot \varepsilon(k,\lambda) = 0$)  where we  defined  for $k \in \R^3 \setminus \{ 0 \}$ and $\lambda \in \Z_2$  the following operator in $\HH$
\begin{align}
& D_1(k,\lambda)  \label{defofD1}  \\
& :=   \frac{  1 }{(2\pi)^{3/2}}  \sum_{j=1}^N \left\{ 2  \frac{ \varepsilon(k,\lambda)  }{\sqrt{2 |k|} } \overline{\kappa(k)} e^{  i k \cdot x_j} \cdot  (p_j +  A(x_j)  ) +     \mu   S_j \cdot \frac{ i k \wedge   \varepsilon(k,\lambda) }{\sqrt{2 |k|} }  \overline{\kappa(k)} e^{  i k  \cdot x_j}\right\}   \nonumber
\end{align}
with domain  $\mathcal{D}((-\Delta + H_{\rm f})^{1/2})$.
Here we understand the right hand side of   \eqref{comm1}  as a Bochner integral, cf. \cite{yosida,evans10,AmannEscher.2009}.
Explicitly,  it follows from \eqref{defofD1} that for   all $\psi \in  \mathcal{D}((-\Delta + H_{\rm f})^{1/2})$  there exists a constant $C$
such that   for almost all $(k,\lambda) \in \R^3 \times \Z_2$
 \begin{align}\label{eq:fundbochest}
 & \|  h(k,\lambda) D_1(k,\lambda) \psi \| = |   h(k,\lambda)  |  \| D_1(k,\lambda) \psi \|  \\
 &  \leq C   | h(k,\lambda) |  \frac{ | \kappa(k)|}{\sqrt{ |k|}}  \sum_{j=1}^N    \left( \| (p_j + A(x) ) \psi \|  )  + |k| \| \psi \| \right)  .  \nonumber
 \end{align}
Since the right hand  side is bounded  in view of the elementary inequalities in Lemma \ref{lem:elemest}, we see from  \eqref{eq:assonfield} that \eqref{eq:fundbochest}
  is  an $L^1$--function  of $(k,\lambda)$, and so the right hand side of   \eqref{comm1} converges for all   $\psi \in \mathcal{D}((-\Delta + H_{\rm f})^{1/2})$  as a Bochner integral
 in $L^1(\R^3 \times \Z_2 ; \HH )$, cf. \cite[V.5. Theorem 1]{yosida} or  \cite[Appendix E.3. Theorem 8]{evans10}. Note that  weak measurability follows directly form the assumptions about $h$ and $\kappa$ and  hence  strong   measurability  is granted  by the separability of  $\HH$, cf.
 \cite[Theorem IV.22]{ReedSimon.1972}.
Taking  a second  commutator, we find for $h \in \hh$ and $\psi \in \mathcal{C}$ that
\begin{align}
     [  a(h) , D_1(k,\lambda) ]   \psi  & = \sum_{\lambda'=1,2} \int_{\R^3}   \overline{ h(k',\lambda')} D_2(k,\lambda,k' , \lambda') \psi  dk' \label{comm2} ,
\end{align}
where we  defined  for $k,k' \in \R^3 \setminus \{ 0 \}$ and $\lambda, \lambda'  \in \Z_2$ the following  bounded  operator in $\HH_{\rm at}$
\begin{align*}
D_2(k,\lambda,k' , \lambda')  & := \frac{ 2 }{(2\pi)^{3}}  \sum_{j=1}^N  e^{  i ( k - k') \cdot x_j}  \overline{  \kappa(k) } \kappa(k')   \frac{  \varepsilon(k,\lambda) \cdot \varepsilon(k',\lambda') }{\sqrt{2|k|} \sqrt{2|k'|}  }   .
\end{align*}
Again we understand   \eqref{comm2} as a Bochner integral in  $L^1(\R^3 \times \Z_2 ; \HH )$, that is for all $(k,\lambda) \in \R^3 \times \Z_2$ and  $\psi \in \HH$
the map   $$(k',\lambda') \mapsto  \overline{h(k',\lambda')} D_2(k,\lambda,k',\lambda') \psi$$ is an element of  $L^1(\R^3 \times \Z_2 ; \HH )$.

Next we formulate a   hypothesis about boundary values of the resolvent.

\begin{hyp} \label{hypresikvebt} Hypothesis \ref{hyp:exgs} holds and  $S \subset \R^3 \setminus \{ 0 \}$  
is  a set such that for all $\lambda, \lambda' \in \Z_2$ and $k, k' \in S$
the boundary value of the resolvent
\begin{align} \label{eq:boundaryvalueres}
& \inn{  D_{1}(k,\lambda) \psi_{\rm gs} , (H - E_{\rm gs} - \omega(k') - i 0_+ )^{-1}  D_{1} (k',\lambda') \psi_{\rm gs} }  \\ \nonumber
 & := \lim_{\eta \downarrow 0} \inn{  D_{1}(k,\lambda) \psi_{\rm gs} , (H - E_{\rm gs} - \omega(k') - i \eta )^{-1}  D_{1} (k',\lambda') \psi_{\rm gs} }
\end{align}
exists. For every compact $S_0 \subset S$ there  exists  an $\epsilon_0 > 0$ such that for all $k \in S$ and $\lambda, \lambda' \in \Z_2$
\begin{align} \label{eq:boundaryvalueres1}
\sup_{k' \in S_0 , \eta  \in (0,\epsilon_0) }|  \inn{  D_{1}(k,\lambda) \psi_{\rm gs} , (H - E_{\rm gs} - \omega(k') - i \eta )^{-1}  D_{1} (k',\lambda') \psi_{\rm gs} }  | < \infty .
\end{align}
\end{hyp}

If  there exits a  ground state,  $k,k' \in \R^3 \setminus \{ 0\}$, and   the limit   \eqref{eq:boundaryvalueres}  exists,  then we can define
\begin{align}
 T(k,\lambda,k',\lambda') & :=    - \inn{   (H + \omega(k')  - E_{\rm gs}   )^{-1} D_{1}(k',\lambda')^* \psi_{\rm gs}          ,  D_{1}(k,\lambda)^*  \psi_{\rm gs} }
\label{defofT}
 \\
&  -
  \inn{  D_{1}(k,\lambda)   \psi_{\rm gs} ,   ( H -  \omega(k') - E_{\rm gs}   - i 0_+  )^{-1} D_{1}(k',\lambda') \psi_{\rm gs} } \nonumber   \\
& +
  \inn{   D_{2}(k,\lambda,k',\lambda')  \psi_{\rm gs} , \psi_{\rm gs} }   ,  \nonumber
\end{align}
 which shall be  called   $T$-matrix. 
 Note that the $T$-matrix depends on the choice of the  ground state, in case where  the ground state energy is degenerate.
We need one more Hypothesis.

\begin{hyp} \label{hypresikvebt2} Hypothesis \ref{hypresikvebt} holds for   $S \subset \R^3 \setminus \{ 0 \}$.
The $T$-matrix \eqref{defofT} is
 bounded  for $k$, $k'$ in any compact subset of $S$.   As a function on $(S \times  \Z_2)^2$ the $T$-matrix depends
for each   $k \in S$   continuously on   $|k'|$ for $k' \in S$.  (c.f. Definition \eqref{rem:contabs}  below).
\end{hyp}

\begin{definition} \label{rem:contabs}   Let $S \subset \R^3$ and   $f : x \mapsto f(x)$ be a function
on $S$. We say that $f(x)$ depends continuously (differentiably) on  $|x|$ if for each $s \in \R^3$ with $|s|=1$
the map
 $r  \mapsto f(r s)$  is a continuous (differentiable) function on the set $\{ r \in [0,\infty ) : r s \in S \}$.

\end{definition}

 In  Section \ref{sec:veri} we  verify Hypothesis  \ref{hypresikvebt2} and  \ref{hypresikvebt}
for an atom with dilation analytic coupling by means of dilation analyticity \cite{BalslevCombes.1971,Simon.1973a}.  We only need that there  exists  a  ground state which
is dilation analytic and whose analytic extension decays exponentially. For atoms with spinless ``electrons''
this assumption has been verified \cite{HaslerLejsek.2022}, see also \cite{OConnor.1972,CombesThomas.1973,ReedSimon.1978,Sigal.2009,GriesemerHasler.2009,HaslerHerbst.2011a} for  related results.
Furthermore we use existence and regularity of analytically dilated  resolvents , which have
 been shown   for atoms in  \cite{JacsicPillet.1995,BachFroehlichSigal.1998b,BachFroehlichSigal.1999,HaslerHerbstHuber.2008,AbouSalemFaupinFrohlichSigal.2009}.

We believe that
alternatively  one could use Mourre's commutator method \cite{Mourre.1980} to verify Hypothesis  \ref{hypresikvebt2} and  \ref{hypresikvebt}.
This theory has been applied  to the standard model of non-relatvistic qed or related models by
various authors   \cite{BachFroehlichSigalSoffer.1999,GeorgescuGerardMoeller.2004,FroelichGriesemerSigal.2008,BonyFaupin.2012}, see also
\cite{BallesterosDeckertHaenle.2020} and references therein, where it has
been applied for a related purpose.
Specifically,  using the limiting absorption principles proven in  these references together
with the existence of an exponentially decaying ground state \cite{BachFroehlichSigal.1998b,GriesemerLiebLoss.2001,Griesemer.2004}.

\begin{theorem}  \label{main:qed} Suppose  Hypotheses
 \ref{thm:fgs01thm4} and  \ref{hypresikvebt2} hold for a set $S \subset \R^3 \setminus \{ 0 \}$.   Then  for  $f, h  \in C_c(S)^2$
we have
\begin{align}
 & \inn{ a^*_{\rm out}(f) \psi_{\rm gs} , a^*_{\rm in}(h) \psi_{\rm gs} }  - \inn{ f , h }  \label{eq:mainform} \\
& = - 2 \pi  i \sum_{\lambda, \lambda'=1,2} \int_{\R^3}  \int_{\R^3} \overline{ f(k , \lambda)} \delta(\omega(k) - \omega(k') )   h(k',\lambda')  T(k,\lambda,k', \lambda')  dk' dk . \nonumber
\end{align}

\end{theorem}

\begin{remark}{\rm  Let us comment on the notation
used in Theorem \ref{main:qed}.
\begin{itemize}
\item[(a)]  For a set $S \subset \R^d$ we shall denote by $C_c(S)$  the set of all continuous
functions $f : S \to \C$ with compact support. Without mention we shall use the canonical isomorphism $C_c(S)^2 \cong \{ f : S \times \Z_2\to \C :  f(\cdot , \lambda) \in C_c(S),\ \lambda=1,2 \}$.
 \item[(b)]

  For $t > 0$  the expression  $\delta(t - \cdot )  $
  is  defined  as a point measure on $(0,\infty)$ at $t$ with mass one.
 Explicitly,  for $f \in C_c(\R^3 \setminus \{ 0 \})$ and $t  > 0$ we define
\begin{align*}
\int_{\R^3}  \delta(t   - |x|) f(x) dx  & :=
\int_{S_2}  \int_0^\infty \delta( t - r )  f(r \omega) r^2 dr dS(\omega)
= \int_{S_2}     f(t  \omega) t^2   dS(\omega) ,
\end{align*}
where $dS$ denotes the  measure on $S_2 := \{ x \in \R^3 : |x| = 1 \}$.
The integral \eqref{eq:mainform} is understood as an iterated integral, where we first integrate with respect to $k'$ according to
the displayed equation just above 
and then with respect to $k$.
\end{itemize}}
\end{remark}

\begin{remark} {\rm
We note that a  relation similar to  \eqref{eq:mainform}   has been shown  in \cite{BallesterosDeckertHaenle.2020}
for the spin boson model. In contrast to the spin boson model  there is in non-relativistic qed
 an additional term in the $T$-matrix,
  which comes from the
 quadratic expressions of the field operators in the Hamiltonian.
}
\end{remark}

\begin{remark}  {\rm Relation  \eqref{eq:mainform}  can be used as a starting point for    calculations of the one photon
 scattering matrix.
 A   procedure
to control the scattering amplitude for arbitrary photon processes up to  remainder terms of arbitrarily high order in the coupling constant   has been carried out in \cite{BachFroehlichPizzo.2007},
using  a type of  iterated perturbation theory. In that paper a relation of the scattering amplitude to Bohr's frequency condition was established.
 In contrast to \cite{BachFroehlichPizzo.2007}  the abstract relation  \eqref{eq:mainform}   allows  to use expansions which have
already been established in the literature. In this regard
we note that  analytic  expansions  of the ground state
have been established in  \cite{GriesemerHasler.2009,HaslerHerbst.2011a,HaslerHerbst.2011b,HaslerLejsek.2022}.
Moreover, expansions of  boundary values of  resolvents have been
	obtained in the literature, see  \cite{BachFroehlichSigal.1998b,BachFroehlichSigal.1999,Sigal.2009,HaslerHerbstHuber.2008} and references therein.
	 }
\end{remark}

\vspace{0.2cm}
\section{Proofs}\label{sec:proofs}

 The basic idea of the proof is to integrate out the time evolution  using the spectral theorem together
with an Abelian limit. For this we shall use the following representation of the asymptotic creation  operator,
which is obtained   by means of  Cook's method.

\begin{lemma} \label{lem:cook}  Suppose  Hypothesis \ref{thm:fgs01thm4} holds for some $E > E_{\rm gs}$.
Then for $\psi = 1_{H \leq E} \psi$ and  $h \in L_\omega^2(\R^3 ; \C^2)$  
\begin{align} \label{eqintform}
a_{\substack{ {\rm out} \\ {\rm in }}}^*(h) \psi
 & = a^*(h) \psi  +  i \int_0^{\pm \infty} \sum_{\lambda=1,2} \int_{\R^3} h(k,\lambda) e^{ i u ( H -  \omega(k) )  } D_1(k,\lambda)  e^{ - i H u} \psi   dk   du ,
\end{align}
where the first   integral over $k$ converges as an integral in $L^1(\R^3 \times \Z_2 ; \HH)$ and the second integral with respect to  $u$ converges in the
sense of  Riemann integrals with respect to the norm  topology in $\HH$.
\end{lemma}

We note that the relation in Lemma \ref{lem:cook} is formally equivalent to the relations given in  \cite[(II.7),(II.8),(II.12),(II.17)]{BachFroehlichPizzo.2007}.
The following proof of Lemma \ref{lem:cook} is based on Hypothesis \ref{thm:fgs01thm4}.

\begin{proof} Let $\psi = 1_{H \leq E}\psi$ and $h \in \hh$.  The convergence of the $k$ integral  in \eqref{eqintform} follows from the following estimate analogous to \eqref{eq:fundbochest}.  Using $\| (H+ i ) e^{ - i H u } \psi \| \leq  ( \max\{ |E_{\rm gs}|,|E|\} + 1 ) \| \psi \| $,
and  \begin{align} \label{elemopbound} \|(\Delta + H_f + 1)^{1/2} ( H+ i )^{-1}     \|   < \infty \end{align} (which follows from \eqref{domprob} and abstract theory of closed
operators in Hilbertspaces \cite[Theorem 5.9]{Weidmann.1980})
we see that there  exists a constant $C$ such that for almost all $(k,\lambda) \in \R^3 \times \Z_2$
\begin{align} \label{estonintbochner}
&\left\| h(k,\lambda) e^{ i u ( H -  \omega(k) )  } D_1(k,\lambda)     e^{ - i H u} \psi  \right\| = \left\| h(k,\lambda) D_1(k,\lambda)     e^{ - i H u} \psi  \right\|  \\
 &  \leq C   | h(k,\lambda) |  \frac{ | \kappa(k)|}{\sqrt{ |k|}}  \sum_{j=1}^N    \left( \| (p_j + A(x) ) (-
\Delta + H_f + 1)^{-1/2} \|  + |k| \right) \| \psi \|   .  \nonumber
 \end{align}
Now  \eqref{estonintbochner}  is  by   \eqref{eq:assonfield} an $L^1$-function of $(k,\lambda)$.
It follows that the integral
\begin{equation} \label{eq:contboch}
\sum_{\lambda=1,2} \int_{\R^3} h(k,\lambda) \left[ e^{ i u ( H -  \omega(k) )  } D_1(k,\lambda)  e^{ - i H u} \right] \psi dk
\end{equation}
 exists  as a Bochner integral.

Next we show   \eqref{eqintform} using a type of  Cook's  argument. Thus let $h \in \hh_\omega$.
 We first claim  that
 the map
\begin{equation} \label{eq:derfun}\eta_h :  \R \to \HH , \quad
 u \mapsto  \eta_h(u) := \left[ e^{ i u H } a^*( e^{- i \omega u } h ) e^{ - i H u} \right] \psi
\end{equation}
is differentiable  with respect to the norm topology in $\HH$. For this,  observe  that $u \mapsto e^{ - i H u} \psi$  as well
as $u \mapsto   e^{ - i H u}  H \psi$ are
 continuously differentiable since $\psi \in  {\rm Ran} 1_{H \leq E}$, as can be seen   
 by  the spectral theorem for  self-adjoint operators.   First, we  assume in addition $\omega h \in \hh_\omega$ and show that
the map $\R \mapsto \HH$ given by $u \mapsto  \xi(u) := a^*( e^{- i \omega u } h ) e^{ - i H u} \psi$
is differentiable. Indeed, for $v \in \R \setminus \{ 0\}$
\begin{align} \label{derofxi}
&  \frac{1}{v} \left( \xi(u+v) - \xi(u)  \right)  \nonumber \\
&  =  \frac{1}{v}  \left( a^*( e^{- i \omega (u+v )} h )  -  a^*( e^{- i \omega u } h )\right)  (H+i)^{-1}  e^{ - i H u} (H+i )  \psi  \nonumber  \\
&  +
 a^*( e^{- i \omega (u+v )} h )  (H+i)^{-1}  \frac{1}{v}  \left(e^{ - i H ( u+v) } (H+i)  \psi -   e^{ - i H u} (H+i) \psi \right) \nonumber \\
& \rightarrow a^*(- i \omega  e^{- i \omega u } h ) e^{ - i H u} \psi + a^*( e^{- i \omega u } h ) e^{ - i H u} ( - i H) \psi  = \xi'(u)
\end{align}
where the convergence  follows using    \eqref{elemopbound} and Lemma \ref{lem:elemest} as well as the differentiability of
$u \mapsto  e^{ - i H u } (H + i ) \psi$, just discussed.
 Thus     $u \mapsto \xi(u)$  is an $\HH$--valued differentiable function. Since  moreover  $\xi(u) \in \mathcal{D}(H)$,  by  Lemma \ref{invdom},
  it  follows
 from 
\begin{align*}  
& v^{-1} ( e^{ i H (u + v ) } \xi(u+v) -   e^{ i H u  } \xi(u) ) =v^{-1} ( e^{ i H (u + v ) } -  e^{ i H u } ) \xi(u) + e^{ i H (u + v ) }   v^{-1} ( \xi(u + v) - \xi(u) )  \\
&  \to   e^{ i H u}   (i H \xi(u) + \xi'(u))    ,
\end{align*}
 that $e^{ i H u } \xi(u)$ is a differentiable function
with derivative $e^{ i H u } (i H \xi(u) + \xi'(u))$.
We now  conclude  using \eqref{derofxi}    that    the function in  \eqref{eq:derfun}
  is for $h \in \hh_\omega$  with $\omega h \in \hh_\omega$ differentiable with derivative
\begin{align}
&  \frac{d}{du} \eta_h(u) = \frac{d}{du}  \left[ e^{ i u H } a^*( e^{- i \omega u } h ) e^{ - i H u} \right] \psi  \label{eq:derfin-100}  \\
& =  e^{ i u H }\left(  [ i H  ,  a^*( e^{- i \omega u } h )] +  a^*(- i \omega e^{- i  \omega u} h) \right)  e^{ - i H u}  \psi  \nonumber  \\
& =   e^{ i u H }   \left[ i  \sum_{j=1}^N  \left\{  ( p_j  +  A(x_j))^2  + \mu S_j \cdot B(x_j) \right\}   ,  a^*( e^{- i \omega u } h ) \right]   e^{ - i H u}  \psi  \nonumber \\
& =   e^{ i u H } \frac{  i }{(2\pi)^{3/2}}    \sum_{j=1}^N   \Bigg\{ \sum_{\lambda=1,2} \int_{\R^3}
 2  \frac{ \varepsilon(k,\lambda)  }{\sqrt{2 |k|} } \overline{\kappa(k)} e^{  i k \cdot x_j}  h(k,\lambda) e^{- i \omega(k) u} dk \cdot  (p_j +  A(x_j)  ) \nonumber
\\
& + \sum_{\lambda=1,2}
 \int_{\R^3}  \mu   S_j \cdot  \frac{ i k \wedge   \varepsilon(k,\lambda) }{\sqrt{2 |k|} }  \overline{\kappa(k)} e^{  i k  \cdot x_j} h(k,\lambda) e^{- i \omega(k) u}  dk
\Bigg\}   e^{ - i H u}  \psi  ,  \label{eq:derfin-10}
\end{align}
 where in the third line   we used the basic relation  $[H_{\rm f}, a^*(f) ] \subseteq  a^*(\omega f)$ and
in the fourth line  we used the canonical commuation relations. Now observe  from the  the explicit  expression
that   \eqref{eq:derfin-10}  depends continuously on $u$ for all $h \in \hh_\omega$.
It now follows from a standard limiting argument, that   the function \eqref{eq:derfun} is  for
all $h \in \hh_\omega$ continuously differentiable in $u$
with derivative given by \eqref{eq:derfin-10}
 (for $h_n = 1_{|\cdot | \leq n} h$, we see from   \eqref{elemopbound}  that  $\eta_{h_n}(u) \to \eta_h(u)$ for each $u \in \R$
and the derivatives $\eta_{h_n}'$ converges uniformly on compact sets in view of \eqref{eq:derfin-10}).

Next we use the definition   of $D_1$ given in   \eqref{comm1} and the convergence of the Bochner-integral justified by  \eqref{estonintbochner}
to see that
\begin{align}
\text{  \eqref{eq:derfin-10} }= e^{ i u H }  \sum_{\lambda=1,2} \int_{\R^3}  h(k,\lambda)  e^{ -i u    \omega(k)   } D_1(k,\lambda)  e^{ - i H u} \psi  dk  , \label{eq:derfin0}
\end{align}
where equality  
 can be seen  by calculating for both sides  the inner product with elements of the
dense subset $\mathcal{C}$  (using  elementary properties of Bochner-integrals, cf.  \cite[V.5. Corollary 2]{yosida}  or \cite[Appendix E.3. Theorem 8]{evans10})  
and interchanging integrals with respect to the integration variables $k$ and $x=(x_1,...,x_N)$
by means of Fubini (in the  l.h.s. one calculates first the $k$-integral and then the $x$-integral and in  the r.h.s. in the reversed order).
 By continuity of  $e^{ i u H }  $  
it follows again
from the convergence  of the Bochner integral, established by means of \eqref{estonintbochner} (using elementary properties of Bochner-integrals %
\cite[V.5. Corollary 2]{yosida} or \cite[Theorem 2.11 (iii)]{AmannEscher.2009})
that for all $h \in  \hh_\omega$
\begin{align} \label{eq:derfin1}
 & e^{ i u H }  \sum_{\lambda=1,2} \int_{\R^3}  h(k,\lambda)  e^{ -i u    \omega(k)   } D_1(k,\lambda)  e^{ - i H u} \psi  dk \nonumber  \\
&=    \sum_{\lambda=1,2} \int_{\R^3}     h(k,\lambda)  e^{ i u  (H -     \omega(k) )   } D_1(k,\lambda)  e^{ - i H u} \psi dk  .
 \end{align}

Finally,
to show \eqref{eqintform} we will  use the definition of the  asymptotic creation and annihilation operators  \eqref{defofasy}
and  \eqref{defofasy2}.
Thus
\begin{align}
 a_{\substack{ {\rm out} \\ {\rm in }}}^*(h) \psi
& = \lim_{t \to \pm \infty} e^{ i H t} a^*( e^{- i \omega t } h)    e^{- i H  t } \psi  \nonumber \\
 & = a^*(h) \psi  +   \int_0^{\pm \infty}  \frac{d}{du}\left[ e^{ i u H } a^*( e^{- i \omega u } h ) e^{ - i H u} \right] \psi du
\nonumber
\\
 & = a^*(h) \psi  + i  \int_0^{\pm \infty}   \sum_{\lambda=1,2} \int_{\R^3}  h(k,\lambda)  e^{ i u ( H -  \omega(k) )  } D_1(k,\lambda)  e^{ - i H u} \psi    dk   du , \label{eq:IIIBNNN}
\end{align}
where for the first identity we used   the existence of the asymptotic creation operators, i.e.,  Hypothesis \ref{thm:fgs01thm4} (i), for the second
identity we used   the fundamental
theorem of calculus  and the continuity of the derivative  in $u$ (i.e. \eqref{eq:derfin-10}),
and in the last identity we used  \eqref{eq:derfin0}  and \eqref{eq:derfin1}. Thus \eqref{eq:IIIBNNN} shows   the desired identity.
We note that instead of using Hypothesis \ref{thm:fgs01thm4} (i)  one could use
a stationary phase estimate  and a density argument to show the integrability at infinity with resect to  $u$  directly, cf.  \cite[Proposition 3, Theorem 4]{FrohlichGriesemerSchlein.2001}
or \cite[Lemma II.1]{BachFroehlichPizzo.2007}.
 \end{proof}

\begin{proposition}  \label{eq:propTmat}  Suppose Hypothesis \ref{thm:fgs01thm4} holds for some $E > E_{\rm gs}$ and that Hypothesis  \ref{hypresikvebt} holds for some $S \subset \R^3 \setminus \{ 0 \}$. Then
for all $k \in S$, $\lambda \in \Z_2 $, and  $h \in C_c(S)^2$
\begin{align*}
&  \inn{      D_{1}(k,\lambda)\psi_{\rm gs}     ,  a_{\rm in}^*(h)  \psi_{\rm gs} }  =  \sum_{\lambda'=1,2}  \int_{\R^3} T(k,\lambda,k',\lambda')  h(k',\lambda') dk'
\end{align*}

\end{proposition}

\begin{proof} To simplify notation we write $K=(k,\lambda)$,   $K' = (k',\lambda')$, and
\begin{align} \label{simplenotation}
    \quad  \int  (\cdots )  dK  = \sum_{\lambda=1,2} \int_{\R^3} (\cdots ) dk .
\end{align}
To calculate  the right  argument of the inner product
we use  Lemma \ref{lem:cook}  and  find  after the  substitution $u \mapsto - u$
\begin{align}
a_{\rm in}^*(h) \psi_{\rm gs}  =  a^*(h)  \psi_{\rm gs}  - i  \int_0^{\infty}    \int  h(K') e^{ -i u ( H -  \omega(k') - E_{\rm gs} )  }  D_1(K')    \psi_{\rm gs} dK'  du \label{step1} .
\end{align}
Next we want to evaluate  the integral in    \eqref{step1} with respect to $u$. For this,  we first observe that
for a self-adjoint operator $A$, we have by the spectral theorem the following identity
$$
\int_0^T e^{- i t (A-i\epsilon)} dt = \frac{1}{-i(A-i\epsilon)} \left( e^{- i T (A - i \epsilon)} - 1 \right) ,   \quad (\epsilon > 0)
$$
where the integral is understood with respect to the strong operator topology as a Riemann-integral.
Taking the limit $T \to \infty$, again in the strong operator topology, gives
\begin{equation} \label{eq:intreppffrac}
\int_0^\infty  e^{- i t (A-i\epsilon)} dt = \frac{-i}{A-i\epsilon}   .
\end{equation}
Let us  apply   $\langle D_1(K)  \psi_{\rm gs} , \cdot \rangle $ to the integral in   \eqref{step1}.
Using that  strong convergence implies weak convergence and  elementary properties of Bochner-integrals 
 \cite[V.5. Corollary 2]{yosida} 
 (see also \cite[Appendix E.3. Theorem 8]{evans10})
we  obtain the following integral. We rewrite this integral
using an Abelian limit, see   Proposition   \ref{lem:limintident} (where the required  existence of the limit follows from the existence
of the limit of the improper integral in \eqref{step1}), we  then  use  \eqref{eq:intreppffrac} to  evaluate  the integral in    \eqref{step1} with respect to $u$.
This gives for  all $k \neq 0$  
\begin{align}
& \inn{    D_1(K)  \psi_{\rm gs} , \int_0^\infty \left( \int  h(K')  , e^{- i u ( H -  \omega(k') - E_{\rm gs} )  }D_1(K') \psi_{\rm gs}  dK' \right)  du }  \nonumber \\
& = \int_0^\infty \left( \int  h(K') \inn{    D_1(K)  \psi_{\rm gs} , e^{- i u ( H -  \omega(k') - E_{\rm gs} )  }D_1(K') \psi_{\rm gs} } dK' \right)  du \nonumber \\
& = \lim_{t \to \infty} \int_0^t \left( \int  h(K') \inn{    D_1(K )  \psi_{\rm gs} , e^{ - i u ( H -  \omega(k') - E_{\rm gs} )  }D_1(K') \psi_{\rm gs} } dK' \right)  du \nonumber  \\
& =  \lim_{\epsilon \downarrow 0 }   \int_0^\infty   e^{- \epsilon u} \left(  \int  h(K')  \inn{    D_1(K )  \psi_{\rm gs} , e^{  - i u ( H -  \omega(k') - E_{\rm gs} )  }D_1(K') \psi_{\rm gs} } dK' \right) du  \nonumber  \\
& =  \lim_{\epsilon \downarrow 0 }   \int_0^\infty   \left(  \int  h(K')  \inn{    D_1(K )  \psi_{\rm gs} , e^{ - i u ( H -  \omega(k') - E_{\rm gs} - i \epsilon  )  }D_1(K') \psi_{\rm gs} } dK' \right) du  \nonumber  \\
& = - i  \lim_{\epsilon \downarrow 0 } \int h(K')   \inn{    D_1(K )  \psi_{\rm gs} , ( H -  \omega(k') - E_{\rm gs} - i \epsilon )^{-1} D_1(K') \psi_{\rm gs} } dK' \nonumber   \\
&= - i  \int  h(K')  \inn{    D_1(K )  \psi_{\rm gs} , ( H -  \omega(k') - E_{\rm gs} - i  0_+  )^{-1} D_1(K') \psi_{\rm gs} } dK'  , \label{step2}
\end{align}
where we first used  Fubini  in the second to last line,
which is justified since the integrand satisfies the bound
\begin{align*}
& \left| h(K')  \inn{    D_1(K )  \psi_{\rm gs} , e^{ - i u ( H -  \omega(k') - E_{\rm gs} - i \epsilon  )  }D_1(K') \psi_{\rm gs} } \right| \\
& \leq |h(K')|  \| D_1(K )  \psi_{\rm gs} \| \| D_1(K') \psi_{\rm gs} \| e^{ - \epsilon u}
\end{align*}
so it is integrable with respect to  $K'$ (by   \eqref{eq:assonfield} and \eqref{eq:fundbochest}) as well as trivially with respect to   $u$,
and  we then used  \eqref{eq:intreppffrac}  in the second to last line.  In the last line  of \eqref{step2} we used dominated convergence, which
is  justified by  \eqref{eq:boundaryvalueres1}  of    Hypothesis \ref{hypresikvebt}.

Next, we shall also use  a pull-through resolvent identity, which states that  for almost all $k' \neq 0$
\begin{equation} \label{eq:pull}
a(k',\lambda') \psi_{\rm gs}  = -(H + \omega(k')  - E_{\rm gs} )^{-1} D_1 (k',\lambda')^* \psi_{\rm gs}  ,
\end{equation}
for a proof  see  for example \cite{Frohlich.1973}. 
Using identity  \eqref{step1}, inserting  \eqref{step2},
and calculating an elementary commutator, we find
\begin{align}
&   \inn{      D_1(K)\psi_{\rm gs}     ,  a_{\rm in}^*(h)  \psi_{\rm gs} } \nonumber  \\
  &  =   \inn{ D_1(K)   \psi_{\rm gs}          ,  a^*(h)   \psi_{\rm gs} } \nonumber  \\
&  \quad -  i
 \int_0^{\infty}  \int h(K')   \inn{    D_1(K)  \psi_{\rm gs} , e^{ - i u ( H -  \omega(k') - E_{\rm gs} )  }D_1(K') \psi_{\rm gs} } dK'  du  \nonumber  \\
  &  =  \inn{ [  a(h) , D_1(K) ]   \psi_{\rm gs}          ,    \psi_{\rm gs} }  + \inn{ a(h)  \psi_{\rm gs}          ,   D_1(K)^* \psi_{\rm gs} } \nonumber  \\
& \quad -  \int  h(K')  \inn{   D_1(K)   \psi_{\rm gs} , ( H -  \omega(k') - E_{\rm gs}  -  i 0_+  )^{-1}  D_1(K') \psi_{\rm gs} }   dK' \nonumber  \\
& =    \int  h(K')\bigg(  \inn{   D_2(K,K')  \psi_{\rm gs} , \psi_{\rm gs} } \nonumber \\
&  \quad   - \inn{   (H + \omega(k')  - E_{\rm gs}  )^{-1} D_1(K')^* \psi_{\rm gs}          ,  D_1(K)^*  \psi_{\rm gs} }
  \nonumber \\ & \quad  -
  \inn{   D_1(K)   \psi_{\rm gs} , ( H -  \omega(k') - E_{\rm gs}  -  i 0_+  )^{-1}  D_1(K') \psi_{\rm gs} }  \bigg) dK' \nonumber  \\
 & =   \int   h(K')  T(K,K') dK'  , \label{eq:proofprop}
\end{align}
where in the second to  last line we used  the definition of $D_2$ given in \eqref{comm2} and \eqref{eq:pull}.
In the
 last line we used  the definition of the $T$-matrix,  \eqref{defofT}.
\end{proof}

\begin{lemma} \label{corscat1}  Suppose Hypothesis \ref{thm:fgs01thm4}  holds for  $E > E_{\rm gs}$. Then  for  $\varphi = 1_{H \leq E} \varphi$ and  $f \in L_\omega^2(\R^3;\C^2)$
we have for all $t \in \R$
$$
e^{ i H t } a_{\rm in}^*(f)  \varphi  =  a_{\rm in}^*(e^{ i \omega t} f )  e^{ i H t} \varphi  .
$$
\end{lemma}

\begin{proof}
For $\psi \in 1_{H \leq E} \HH$ we find using  \eqref{defofasy}, \eqref{defofasy2}, and  Hypothesis \ref{thm:fgs01thm4} (i)
\begin{align*}
 e^{- i H t} a_{\rm in}^*(e^{ i \omega t } f) e^{ i H t } \psi
 & =e^{- i H t}  \lim_{u \to -\infty} e^{ i H u}  e^{- i H_{\rm f} u }  a^*(e^{ i \omega t  } f)  e^{ i H_{\rm f}  u} e^{- i H  u } e^{ i H t } \psi    \\
  & = \lim_{u \to -\infty}  e^{ i H (u-t) }  e^{- i H_{\rm f} (u-t) }  a^*( f)  e^{ i H_{\rm f} ( u-t)} e^{- i H  (u-t) } \psi    \\
& =   a_{\rm in}^*( f) \psi .
\end{align*}
Multiplying both sides with $e^{ i H t}$ yields the claimed identity.
\end{proof}

\begin{proof}[Proof of Theorem  \ref{main:qed} ]
From Lemma \ref{lem:cook} we see that for   the ground state  $\psi_{\rm gs}$  of $H$ with ground state energy $E_{\rm gs}$ we find
\begin{align} \label{eq:scatrel2}
a_{\rm out}^*(f)  \psi_{\rm gs} - a_{\rm in}^*(f) \psi_{\rm gs}   = i \int_{-\infty}^\infty   \sum_{\lambda=1,2} \int_{\R^3} f(k,\lambda)   e^{ i u (H -   \omega(k) - E_{\rm gs}) } D_1(k,\lambda)  \psi_{\rm gs} dk du .
\end{align}
We shall use notation \eqref{simplenotation}.
Using Hypothesis \ref{thm:fgs01thm4} (ii),  \eqref{eq:scatrel2} together with the continuity if the inner product,  Proposition \ref{lem:limintident},  and  elementary properties  of Bochner-integrals,   \cite[V.5. Corollary 2]{yosida},  
we find
\begin{align*}
& L := \inn{ a_{\rm out}^*(f) \psi_{\rm gs} , a_{\rm in}^*(g) \psi_{\rm gs} } - \inn{ f , g } \\
&  = \left\langle (  a_{\rm out}^*(f)  - a_{\rm in}^*(f) ) \psi_{\rm gs} , a_{\rm in}^*(g) \psi_{\rm gs} \right\rangle +  \inn{  a_{\rm in}^*(f)  \psi_{\rm gs} , a_{\rm in}^*(g) \psi_{\rm gs} } - \inn{ f , g }  \\
& = -i \lim_{t \to \infty} \int_{-t}^t \inn{  \int  f(K)   e^{ i u (H - \omega(k) - E_{\rm gs}) }  D_1(K) \psi_{\rm gs}   dK  ,    a_{\rm in}^*(g) \psi_{\rm gs} }du  \\
& = -i \lim_{\epsilon \downarrow 0 } \int_{-\infty}^\infty  e^{-\epsilon |u|}  \left\langle  \int f(K)   e^{ i u (H - \omega(k) - E_{\rm gs}) }  D_1(K) \psi_{\rm gs}   dK   ,    a_{\rm in}^*(g) \psi_{\rm gs} \right\rangle du \\
& = - i  \lim_{\epsilon \downarrow 0 } \int_{-\infty}^\infty e^{-\epsilon |u|}  \int \overline{f(K)}   \inn{  D_1(K) \psi_{\rm gs}     ,   e^{ - i u (H - \omega(k) - E_{\rm gs}) }  a_{\rm in}^*(g) \psi_{\rm gs} } dK du  .
\end{align*}
Now using  first Lemma \ref{corscat1} and then Proposition  \ref{eq:propTmat},     we  obtain
\begin{align*}
i  L  & =   \lim_{\epsilon \downarrow 0 } \int_{-\infty}^\infty  e^{-\epsilon |u|}  \int \overline{f(K)}   \inn{  D_1(K) \psi_{\rm gs}     ,   e^{  i u \omega(k)  }  a_{\rm in}^*(e^{- i u \omega} g) e^{ - i u (H  - E_{\rm gs}) } \psi_{\rm gs} } dK du \\
& =  \lim_{\epsilon \downarrow 0 } \int_{-\infty}^\infty  e^{-\epsilon |u|} \int \overline{ f(K)} e^{  i u \omega(k)  }   \inn{  D_1(K) \psi_{\rm gs}     ,    a_{\rm in}^*(e^{- i u \omega} g)   \psi_{\rm gs} }dK du \\
& =  \lim_{\epsilon \downarrow 0 } \int_{-\infty}^\infty  e^{-\epsilon |u|}  \int   \int \overline{f(K)}  g(K')   e^{ i u (\omega(k)  - \omega(k') ) } T(K,K')  dK' dK du \\
& =  \lim_{\epsilon \downarrow 0 }  \int  \int \overline{f(K)}  g(K') \int_{-\infty}^\infty  e^{-\epsilon |u|}   e^{ i u (\omega(k)  - \omega(k') ) } T(K,K') du dK' dK  \\
& =  \lim_{\epsilon \downarrow 0 }  \int \int \overline{f(K)}  g(K') \frac{2 \epsilon}{\epsilon^2 + (\omega(k)  - \omega(k') )^2} T(K,K') dK'  dK \\
& =  2 \pi \int \int \overline{ f(K)} g(K')   \delta(  \omega(k')  - \omega(k) ) T(K,K') dK' dK  ,
\end{align*}
where in the fourth   line we used Fubini 
 (which is justified by  the boundedness assumptions in Hypotheses \ref{hypresikvebt} and \ref{hypresikvebt2}). In the last equality we integrated by Hypothesis \ref{hypresikvebt2} first over
 $|k'|$  and then use  dominated convergence for the  $k$-integration  to take the limit $\epsilon \downarrow 0$ into the integal.
\end{proof}

\section{Verification  of the Hypothesis} 

\label{sec:veri}

In this section we show, in Theorem  \ref{propverihypD},  that  Hypothesis \ref{hypresikvebt2} (and thus Hypothesis \ref{hypresikvebt})     holds for an atom
for small  but nontrivial coupling in an energy interval,  for which   a Fermi's golden rule condition holds.
We make the following  additional assumptions on the potential $V$, which are assumed to hold throughout this section.
Let
$$
V(x_1,....,x_N) = V_C(x_1,...,x_N) :=  -  \sum_{j=1}^N  \frac{Z \alpha  }{|x_j|}  +  \sum_{i < j }^N  \frac{\alpha}{|x_j - x_i|}  ,
$$
with $Z=N$ and $\alpha > 0 $.  Let  $\HH_{\rm at} =L^2_a((\R^3 \times \Z_{2s+1} )^N) $ be given as  the  antisymmetric subspace.
Then it follows  that the spectrum of $H_{\rm at} = - \Delta + V_C$  has the structure
\begin{align} \label{stucspec}
\sigma(H_{\rm at}) = \{ E_{{\rm at}, j} : j=0,..., M \} \cup [\Sigma_{\rm at}, \infty ) ,
\end{align}
where   $\Sigma := \inf \sigma_{\rm ess}(H_{\rm at})$, $M \in \{1,3,...\} \cup \{ \infty \}$,  $(E_{{\rm at}, j})_{j=1,...,M}$ is strictly monotone, and $E_{{\rm at},j} < \Sigma$, see \cite{ReedSimon.1978} and references therein.
For $\theta \in \R$, $\psi \in \HH_{\rm at}$ and $h \in \hh$  we define the transformations
\begin{align*}
U_{\rm el}(\theta) \psi(x_1,s_1,...,x_N,s_N)  & = e^{  \frac{3 \theta}{2} N}  \psi(e^\theta x_1,s_1,..., e^\theta x_N,s_N) , \quad x_j \in \R^3 ,  s_j \in \{1, 2s+1\}  \\
U_{\rm ph}(\theta) h(k,\lambda)  & = e^{-\frac{3 \theta}{2}}  h(e^{-\theta} k , \lambda ) , \quad (k,\lambda) \in \R^3 \times \Z_2 .
\end{align*}
Note that it is straight forward to see that these transformations are unitary. Let $\Gamma(U_{\rm ph})$ denote the bounded linear operator
on  $\FF(\hh)$ such that $\Gamma(U_{\rm ph})|_{\FF_n(\hh)} = \bigotimes_{j=1}^n U_{\rm ph}$.
We define $U(\theta) = U_{\rm el} \otimes \Gamma(U_{\rm ph})$. 
As an immediate consequence of the definitions one finds
\begin{align*}
U(\theta)  x_j U(\theta)^{-1}   & =  e^\theta x_j , \quad
U(\theta)  p_j U(\theta)^{-1}    =   e^{-\theta} p_j , \quad
 U(\theta)  H_{\rm f} U(\theta)^{-1}   =  e^{-\theta} H_{\rm f} .
\end{align*}
For the Hamiltonian $H$, given in \eqref{eq:defofham}, we define for $\theta \in \R$
$$
H(\theta) = U(\theta) H U(\theta)^{-1} .
$$
A straight forward calculation shows that
 \begin{equation} \label{eq:defofhamthet}
H(\theta)   = \sum_{j=1}^N  \{ (e^{-\theta} p_j + A_\theta(x_j))^2 +   \mu S_j \cdot B_\theta(x_j) \}  + e^{-\theta}  V_C \otimes \one  + e^{-\theta}  \one \otimes H_{\rm f} ,
\end{equation}
where  we defined
$$
A_{\theta,l}(x)  := \phi( G_{x,l,\theta} ) ,  \quad  B_{\theta,l}  := \phi(J_{x,l,\theta}) , \quad x \in \R^3 , \quad l=1,2,3 ,
$$
with
\begin{align*}
G_{x,l,\theta}(k,\lambda)  & := \frac{e^{-\theta} }{(2\pi)^{3/2}} \frac{\kappa(e^{-\theta} k)}{\sqrt{|k|}} \varepsilon_l(k,\lambda) e^{-ik \cdot x} , \\  J_{x,l,\theta}(k,\lambda)  & :=\frac{ e^{-\theta}}{(2\pi)^{3/2}}   \frac{\kappa( e^{-\theta} k)}{\sqrt{|k|}}  [ (-i e^{-\theta}   k)  \wedge \varepsilon(k,\lambda) ]_l e^{-ik \cdot x} ,
\end{align*}
for $k \in \R^3 \setminus \{ 0 \}$ and $\lambda \in \Z_2$.
We   assume that the coupling function is of the form
\begin{align} \label{modelasskappa}
\kappa(k)  = g \tilde{\kappa}(|k|) ,\quad  k \in \R^3 ,
\end{align}
 with $g \in \R$,  called minimal coupling constant, and with
$\tilde{\kappa} : (0,\infty)  \to \R$   a   positive  function,  which  satisfies   $\tilde{\kappa}(r) \to a$ as $r \to \infty$ for some $a \in [0,\infty)$, and has an analytic
continuation to a cone,     $\{ r  e^{i \varphi} :  r > 0 ,  - \tilde{\theta}_0 <   \varphi  <  \tilde{\theta}_0 \}$ for some $\tilde{\theta}_0  > 0$,
around the positive real axis, which is bounded and decays faster than any inverse polynomial, e.g. $\tilde{\kappa}(r) =  \exp(-r^4)$.
The Hamiltonian $H(\theta)$ is self-adjoint on the domain   $ \mathcal{D}(-\Delta + H_{\rm f})$ for any $\theta \in \R$, see Remark  \ref{rem:def}.
 For  $a \in \C$ and $ r \geq 0$  we define $D_r(a) :=\{ z \in \C : |z-a| < r \}$ and $D_r := D_r(0)$.

The  following lemma collects a elementary facts, which can be shown using elementary estimates.
A proof can be found for example in \cite[Lemma 1.1,Corollary 1.4]{BachFroehlichSigal.1999}
or \cite{HaslerLejsek.2022}.

\begin{lemma} \label{lem:anaext} Suppose $V= V_C$ and  \eqref{modelasskappa}.
Let $\theta_0 \in (0,\min\{ \tilde{\theta}_0,\pi/4 \})$.
  Then the following
holds. \begin{itemize}
\item[(a)]  There exists a $g_0> 0$  such that for all $g \in (-g_0,g_0)$   the mapping $\theta \mapsto H(\theta)$ has an analytic
continuation  to  $D_{\theta_0}$. The resulting analytic continuation
is an  analytic family  of type (A) with common domain of the operators $\mathcal{D}(-\Delta + H_{\rm f})$.
\item[(b)]  The maps  $\theta \mapsto A_\theta(x_j) (-\Delta + H_{\rm f} + 1)^{-1/2}$ and
  $\theta \mapsto     \overline{(-\Delta + H_{\rm f} + 1)^{-1/2}  A_\theta(x_j) |}^{\rm cl}_{\mathcal{D}((H_{\rm f} + 1)^{1/2})}$
are analytic functions on $D_{\theta_0}$.
\item[(c)] Let $W_g(\theta) := \sum_{j=1}^N \left\{   A_\theta(x_j) \cdot  e^{-\theta}  p_j + e^{-\theta}  p_j \cdot
A_\theta(x_j) +  A_\theta(x_j)^2   + \mu S_j \cdot B_\theta(x_j)  \right\}$. There exists a constant $C$ such that for all  $\theta \in D_{\theta_0}$ and $g \in \R$
\begin{align*}
\| W_g(\theta)(-\Delta + H_f + 1 )^{-1}  \|  \leq C |g|(1+|g|) .
\end{align*}
\end{itemize}
\end{lemma}

Henceforth,  we shall denote by $H(\theta)$ the analytic continuation  to a neighborhood of zero  granted by Part~(a) of Lemma~\ref{lem:anaext}.

 \begin{lemma}  \label{remeasyspecdef} Suppose $V= V_C$ and  \eqref{modelasskappa}.
  Let $\theta_0 \in (0,\min\{ \tilde{\theta}_0,\pi/4 \})$ and  $x \in \R$. Then there exists a
$y > 0$ and a  $g_0 > 0$ such that
\begin{align} \label{resest}
x + i y \in \rho(H(\theta))
\end{align}
 for all $\theta \in D_{\theta_0}$ with ${\rm Im} \theta \leq 0$  and $g \in (-g_0,g_0)$.
\end{lemma}
\begin{proof}  We devide the proof into steps. \\
\underline{Step 1:}
Since $V_C$ is infinitesimally $-\Delta$ bounded \cite{ReedSimon.1975},  we have
\begin{align}
\lim_{y \to \infty} \| V_C (  -\Delta  -  i y )^{-1} \|  = 0 .
\end{align}

\vspace{0.3cm}

\noindent
\underline{Step 2:}
Let $x \in \R$  and $y_0  = \sqrt{x^2 +1} $. Then
\begin{align*}
\sup_{\substack{ \theta \in D_{\theta_0}\\  {\rm Im} \theta \leq 0  } } \sup_{y \geq y_0 }
 \| ( -\Delta-  i y ) ( - e^{-2 \theta} \Delta  + e^{-\theta} H_f -  x - i y )^{-1} \| < \infty ,
\end{align*}
\begin{align*}
\sup_{\substack{ \theta \in D_{\theta_0}\\  {\rm Im} \theta \leq 0  } } \sup_{y \geq y_0 } \|
 (-\Delta + H_f + 1 ) ( - e^{-2 \theta} \Delta  + e^{-\theta} H_f -  x -  i y )^{-1} \| < \infty  .
\end{align*}
These bounds follow from the bounds \eqref{esofrespos0}--\eqref{esofrespos2}  below and the triangle inequality.
Let $\theta \in D_{\theta_0}$ with ${\rm Im} \theta \leq 0$.
Then using the spectral theorem and  ${\rm Im}  (  e^{- 2 \theta} r   +  e^{-\theta}  s + x )  \leq 0 $ for $r,s \in [0,\infty)$,
we obtain  for $y \geq y_0$ the following estimates.
We  estimate
\begin{align}\label{esofrespos0}
& \left\| y  ( - e^{- 2 \theta} \Delta  +  e^{-\theta}  H_f  -  x  -  i y )^{-1} \right\| \\
&  =    \sup_{r,s \geq 0 } \left| \frac{   y }{
 e^{- 2 \theta} r   +  e^{-\theta}  s - x - i y } \right| \leq \frac{y}{ \left|  {\rm Im}  (  e^{- 2 \theta} r   +  e^{-\theta}  s + x )   - y  \right|   } \leq   \frac{y}{   y   }  =1 ,\nonumber
\end{align}
and
\begin{align}\label{esofrespos1}
& \left\|-\Delta  ( - e^{- 2 \theta} \Delta  +  e^{-\theta}  H_f  -  x  -  i y )^{-1} \right\|   =    \sup_{r,s \geq 0 } \frac{ r   }{ \left|
 e^{- 2 \theta} r   +  e^{-\theta}  s -  x -  i y \right| }  \\
& \leq   \sup_{r,s \geq 0 } \frac{ r   }{
\left| ({\rm Re} ( e^{- 2 \theta} r   +  e^{-\theta}  s + x ))^2 +  y^2   \right|^{1/2}}  \leq
 \sup_{r,s \geq 0 }  \frac{ r   }{
\left| \frac{1}{2} ({\rm Re} ( e^{- 2 \theta} r   +  e^{-\theta}  s) )^2 +  y^2  - x^2   \right|^{1/2}  } \nonumber \\
& \leq   \sup_{r,s \geq 0 }  \frac{ r   }{
\left| \frac{1}{2} ({\rm Re} ( e^{- 2 \theta} r )^2 + 1 \right|^{1/2}  }  \leq   \sup_{r,s \geq 0 }  \frac{ r   }{
\left|\frac{1}{2} e^{-4 \theta_0}  \cos^2( 2 \theta_0) r^2   +  1  \right|^{1/2}  }  \leq  \frac{ \sqrt{2} e^{2 \theta_0}}{ \cos(2 \theta_0)} ,  \nonumber
\end{align}
 where in the second inquality  we used  $(u+v)^2 \geq \frac{1}{2} u^2 - v^2 $.
An analogous estimate  as \eqref{esofrespos1} (interchanging the roles of $s$ and $r$)  gives
\begin{align}\label{esofrespos2}
& \left\| H_f ( - e^{- 2 \theta} \Delta  +  e^{-\theta}  H_f  -  x  -  i y )^{-1} \right\|   =    \sup_{r,s \geq 0 }  \frac{ s }{ \left|
 e^{- 2 \theta} r   +  e^{-\theta}  s -  x -  i y\right|   } \leq  \frac{ \sqrt{2} e^{ \theta_0}}{ \cos( \theta_0)} .
\end{align}

\vspace{0.3cm}

\noindent
\underline{Step 3:}  It follows from Lemma \ref{lem:anaext} (c) and Step 1 and Step 2, that there exists a
$g_0 > 0$ and a $y > 0$  such that for all $\theta \in D_{\theta_0}$ with ${\rm Im} \theta \leq 0$
and $g \in (-g_0,g_0)$ we have
\begin{align}
\| ( e^{-\theta} V_C + W_g(\theta) ) ( - e^{-2 \theta} \Delta + e^{-\theta} H_f - x - i y )^{-1} \| \leq 1/2 .
\end{align}
It follows using a Neumann expansion for the operator
 $$H(\theta) - x - i y =  ( - e^{-2 \theta} \Delta + e^{-\theta} H_f  - x - i y ) + ( e^{-\theta} V_C + W_g(\theta) )$$  that $x + i y \in \rho(H(\theta))$
 for all $\theta \in D_{\theta_0}$ with ${\rm Im} \theta \leq 0$
and $g \in (-g_0,g_0)$.

\end{proof}

For $x \in \R^n$ we write $\langle x \rangle = (  1 + x^2 )^{1/2} $.
The proof 
 of Theorem  \ref{propverihypD}, below,   is  based on the following nontrivial result, which is stated in the  following  Hypothesis. 

\begin{hyp}  \label{GSanahyp00}  The ground state $\psi_{\rm gs}$ of $H$ has the following properties. There exist positive $\beta > 0$ and
 $\theta_0 > 0$ such that
\begin{itemize}
\item[(i)] $\theta \mapsto \psi_{{\rm gs},\theta} := U(\theta) \psi_{\rm gs}$ has an $\HH$--valued  analytic extension for $\theta$ into $D_{\theta_0}$,
\item[(ii)]  $\|  e^{ \beta \langle \cdot \rangle  } \psi_{{\rm gs},\theta}\| $  is uniformly bounded  on $D_{\theta_0}$.
\end{itemize}
\end{hyp}

\begin{remark} \label{GSanahyp002} {\rm   We will use the following simple consequence of  Hypothesis  \ref{GSanahyp00}   that $\theta \mapsto  e^{ \beta \langle \cdot \rangle  } \psi_{{\rm gs},\theta}$
is an $\HH$--valued analytic function. This follows as an application  of  the abstract result in Lemma   \ref{thm:ana2}.   }
\end{remark}

Let us now introduce notation to formulate
Fermi's golden rule condition. Let $P_{{\rm at},j} = 1_{\{E_{{\rm at}, j} \}}(H_{\rm at})$, where $1_A(x) := 1$ if $x \in A$ and otherwise $1_A(x) = 0$.
Define   for $ k \in \R^3 \setminus \{ 0 \}$ and $ \lambda \in \Z_2$   the operator
$$
w(k,\lambda)  := \sum_{j=1}^N \left\{ 2 G_{x_j}(k,\lambda) \cdot p_j + \mu S_j \cdot J_{x_j}(k,\lambda) \right\}
$$
in $\HH_{\rm at}$ on the domain $D((-\Delta)^{1/2})$ and the linear map
\begin{align}
Z_j  & := \lim_{\epsilon \downarrow 0} \sum_{\lambda=1,2} \int_{\R^3} P_{{\rm at},j} w(k,\lambda) \overline{P}_{{\rm at},j}
\left(  \overline{P}_{{\rm at},j} H_{\rm at} - E_{{\rm at}, j}+ |k| - i \epsilon \right)^{-1}  \overline{P}_{{\rm at},j}  w(k,\lambda)^* P_{{\rm at},j} dk \nonumber \\
& +  \sum_{\lambda=1,2} \int_{\R^3} P_{{\rm at},j} w(k,\lambda) {P}_{{\rm at},j}  w(k,\mu)^* P_{{\rm at},j} \frac{dk}{|k|}  ,
\label{defofZj}
\end{align}
where  $\overline{P}_{{\rm at},j} =\one_{\HH_{\rm at}}  -  {P}_{{\rm at},j}$.
We say that a Fermi's golden Rule condition holds  for the $j$-the  eigenvalue, $E_j$, if
\begin{align} \label{fermigoldenrule}
  {\rm Im} Z_j   > 0  .
\end{align}
Let us state  the following theorem, which follows directly from the results in   \cite{BachFroehlichSigal.1999,HaslerHerbstHuber.2008}.
To formulate it, we define
 \begin{align} \label{defofdeltaj} \delta_j := \dist(E_{{\rm at}, j} , \sigma(H) \setminus \{ E_{{\rm at}, j} \} ) \end{align} and
\begin{equation}
\mathcal{A}(E, \delta,c)  :=  ( E - \delta/2 , E + \delta/2 )+ i [ - c  g^{2}, \infty)  .
\end{equation}

\begin{figure}[h]
\begin{center}
\begin{tikzpicture}[scale=1]
\draw[white, semitransparent, fill=green!10]  (2.8,-0.3) rectangle (4.7,3) ;
\draw[green]  (2.8,-0.3) --  (2.8,3);
\draw  (3.75,1) node[above] {$\mathcal{A}(E_{{\rm at}, j},\delta_j,c_j)$};
\draw[green]  (2.8,-0.3) --  (4.7,-0.3);
\draw[green]  (4.7,-0.3) --  (4.7,3);
\draw[blue!30, semitransparent, fill=blue!30] (-2,0)  --  (0,-2) -- (2,-2);
\draw[blue!30, semitransparent, fill=blue!30] (5.6,-0.6)  --  (7.6,-2) -- (9.6,-2);
\draw[blue!30, semitransparent, fill=blue!30] (-0.3,-0.5)  --  (1.7,-2) -- (3.7,-2);
\draw[blue!30, semitransparent, fill=blue!30] (1.2,-0.5)  --  (3.2,-2) -- (5.2,-2);
\draw[blue!30, semitransparent, fill=blue!30] (3.2,-0.5)  --  (5.2,-2) -- (7.2,-2);
\draw[blue!30, semitransparent, fill=blue!30] (3.6,-0.6)  --  (5.6,-2) -- (7.6,-2);
\draw[blue!30, semitransparent, fill=blue!30] (5.2,-0.5)  --  (7.2,-2) -- (9.2,-2);
 \draw[-] (-4,0) -- (0.5,0) ;
  \draw[thick, dotted] (0.6,0) -- (0.9,0) ;
  \draw[-] (1,0) -- (7,0) ;
    \draw[thick, dotted] (7.1,0) -- (8,0) ;
      \draw[->] (8.1,0) -- (9,0) coordinate (x axis);
\draw (-1,-0.1) -- (-1,0.1) node[above]{$E_{{\rm at}, 0}$};
\draw[->, thick] (-1.2,0.2) -- (-1.8,0.2) ;
\filldraw[blue!80]  (-2,0) circle (2pt)  node[left] {$E_{{\rm gs}}$} ;
\draw (0.3,-0.1) -- (0.3,0.1) ;
\draw[->, thick] (0.2,-0.1) -- (-0.2,-0.4) ;
\filldraw[blue!80]  (-0.3,-0.5) circle (2pt)   ;
\draw (1.8,-0.1) -- (1.8,0.1) node[above]{$E_{{\rm at}, j-1}$} ;
\draw[->, thick] (1.7,-0.1) -- (1.3,-0.4) ;
\filldraw[blue!80]  (1.2,-0.5) circle (2pt)   ;
\draw (3.8,-0.1) -- (3.8,0.1) node[above]{$E_{{\rm at}, j}$};
\draw[->, thick] (3.7,-0.1) -- (3.3,-0.4) ;
\filldraw[blue!80]  (3.2,-0.5) circle (2pt)  node[left] {$E_{ j,1}$} ;
\draw[->, thick] (3.8,-0.1) -- (3.6,-0.5) ;
\filldraw[blue!80]  (3.6,-0.6) circle (2pt) node[right] {$E_{ j,2}$} ;
\draw (5.8,-0.1) -- (5.8,0.1)  node[above]{$E_{{\rm at}, j+1}$}  ;
\draw[->, thick] (5.7,-0.1) -- (5.3,-0.4) ;
\filldraw[blue!80]  (5.2,-0.5) circle (2pt)  ;
\draw[->, thick] (5.8,-0.1) -- (5.6,-0.5) ;
\filldraw[blue!80]  (5.6,-0.6) circle (2pt)  ;
 \end{tikzpicture}
\end{center}
\caption{\small Illustration of the situation in Theorem \ref{bfs99thm3.2}, where the putative spectrum  of $H(\theta)$
is indicated with blue. }
\end{figure}
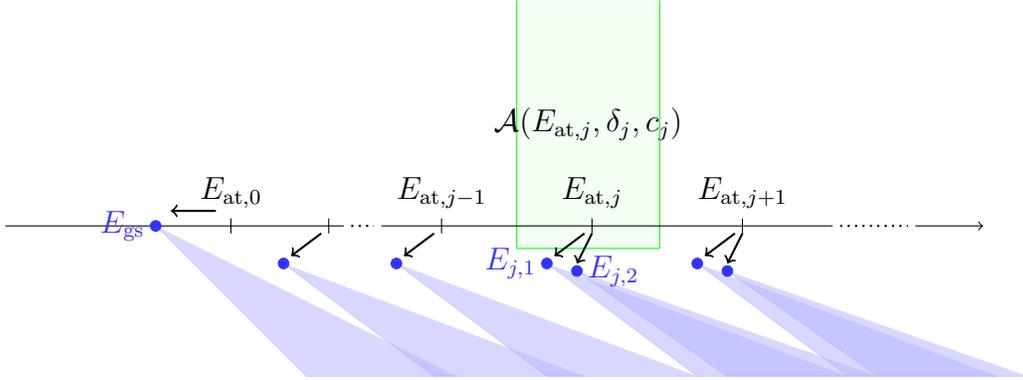

\begin{theorem} 
\label{bfs99thm3.2} Suppose $V= V_C$ and  \eqref{modelasskappa}.
Let $\theta_0 \in (0,\min\{ \tilde{\theta}_0,\pi/4 \})$ and
  $j \in \{1,...,N\}$.  Assume
    ${\rm Im} Z_j  > 0$.  Suppose  $\theta = i \vartheta$ with  $\vartheta \in (0, \min\{ \theta_0, \pi/4\})$.
Then there exist   $c_j > 0$ and  $g_0 > 0$ such that for all $g \in \R$ with $0 < |g| < g_0$
$$
\mathcal{A}(E_{{\rm at}, j}, \delta_j ,c_j)  \subset \rho(H(\theta) )  ,
$$
in particular the interval $[E_{{\rm at}, j} - \delta_j/2, E_{{\rm at}, j} + \delta/2]$ is contained in the resolvent set $\rho(H(\theta) )$,
and there exists a constant $C_j$  such that for any $z \in  \mathcal{A}(E_{{\rm at }, j}, \delta_j ,c_j) $  we have
\begin{align*}
\| H_g(\theta) -z )^{-1} \|  \leq C_j .
\end{align*}
\end{theorem}
\begin{proof}
The theorem follows directly from   \cite[Corollary 8]{HaslerHerbstHuber.2008}  (or  \cite[Theorem 3.2]{BachFroehlichSigal.1999} which needs an additional
non-degeneracy condition) by observing that the set given there contains  for some $c_j > 0$ the set  $\mathcal{A}(E_{{\rm at}, j}, \delta_j ,c_j)$ for all small nonzero  $|g|$.
The result \cite[Corollary 8]{HaslerHerbstHuber.2008} is formulated for electrons with spin $s=1/2$, but the proof carries over to the spinless case $s=0$ (by dropping the additional term which couples to the spin).
 The main idea of the proof  of  \cite[Corollary 8]{HaslerHerbstHuber.2008} 
 is to consider $H(\theta)|_{g=0}$ and to use a perturbation expansion   in $g$. More explicitly, one studies a so called  Feshbach projection with respect to an  energy interval around  $E_j$ and uses  Fermi's golden rule condition.
\end{proof}

We are now ready to prove the main   theorem of this section.

\begin{theorem} \label{propverihypD} Suppose $V= V_C$,   \eqref{modelasskappa} and   that Hypothesis \ref{GSanahyp00} holds. Let $j \in \{1,..., M \}$. Suppose ${\rm Im} Z_j  > 0$.
Then there exists a $g_0 > 0$ such that for all real $g$ with $0< |g| <  g_0$ Hypothesis  \ref{hypresikvebt2}  holds for the Hamiltonian $H$ and  the set $$S_j  = \{ k \in \R^3 \setminus \{ 0 \} :  E_{{\rm at}, j} - E_{\rm gs} -  \delta_j/2  <  \omega(k)  <  E_{{\rm at}, j} - E_{\rm gs}  +  \delta_j/2 \} . $$
In fact, the $T$-matrix is   for $k, k' \in S$  a $C^\infty$-function of $|k|$ and $|k'|$.
\end{theorem}

For the  proof of  Theorem  \ref{propverihypD}     we will use  the following interpolation result.

\begin{lemma}  \label{symop2} Suppose $V= V_C$ and  \eqref{modelasskappa}.
Let $\theta_0 \in (0,\min\{ \tilde{\theta}_0,\pi/4 \})$ and
let $g_0 > 0$ be such that
the assertions  of Lemma \ref{lem:anaext}  hold.  Let $|g| < g_0$ and let
 $U \subset D_{\theta_0}  \times \C$ be an open set such that $ (H(\theta) - z )^{-1}$ is an analytic
 bounded operator valued   function of  $(\theta,z) \in U$. Then  for each $(\theta, z) \in U$ the closure of the   operator
\begin{align} \label{symop}
&  ( - \Delta+ H_{\rm f} + 1)^{1/2}    (H(\theta) - z )^{-1} ( -\Delta + H_{\rm f} + 1)^{1/2}
\end{align}
is bounded and depends analytically on $(\theta,z)$ in $U$.
 \end{lemma}

 \begin{proof}  First we show  that the operators $(H(\theta) - z )^{-1}  ( -\Delta + H_{\rm f} + 1)$  with domain   $\mathcal{C}$  depend continuously on $(\theta,z) \in U$, w.r.t. the  operator norm topology.
  For this let $(\theta',z') \in U$.  Since $H(\theta)$ is an analytic family of type (A), by Lemma  \ref{lem:anaext},  it follows that
  $$
  T_0(\theta) :=  (H(\theta') - z')^{-1} (H(\theta') - H(\theta) )
  $$
  is a bounded operator the domain $\mathcal{D}(-\Delta + H_{\rm f})$   and    $\lim_{\theta \to \theta'} \lim T_0(\theta) = 0 $ (see for example   \cite[VII. \S 2  (2.4)]{Kato.1980}). Thus for $(\theta,z)$ sufficiently close to $(\theta',z')$ we  obtain   from  a Neumann
  expansion the following identity   on $\mathcal{C}$
 \begin{align}
 &  (H(\theta) - z )^{-1}  ( -\Delta + H_{\rm f} + 1)  \nonumber \\
  & =  (H(\theta) - z )^{-1} (H(\theta) - z'  )  (H(\theta) -  z'  )^{-1}   ( -\Delta + H_{\rm f} + 1) \nonumber \\
& =   ( 1 + (z - z'  )  (H(\theta) - z )^{-1} )  (H(\theta') -  z'  - (  H(\theta')-H(\theta) )   )^{-1}   ( -\Delta  + H_{\rm f} + 1)  \nonumber \\
& =   ( 1 + (z - z'  )  (H(\theta) - z )^{-1} ) \sum_{n=0}^\infty  T_0(z)^n    (H(\theta') -  z'  )^{-1}    ( -\Delta + H_{\rm f} + 1) .  \label{opcont}
\end{align}
Continuity  of  \eqref{opcont}  in $(\theta,z)$ now follows, since  $ (H(\theta') -  z'  )^{-1}    ( -\Delta + H_{\rm f} + 1) $ is bounded  (for the latter  observe that $H(\theta')$ is a closed operator with the same domain as
the closed operator $-\Delta + H_{\rm f}$, c.f. \cite[Theorem 5.9]{Weidmann.1980}).
Likewise one  shows that  $( -\Delta + H_{\rm f} + 1) (H(\theta) - z )^{-1} $ is a bounded operator depending continuously on $(\theta,z)$.
It now follows from interpolation, explicitly  Lemma  \ref{intpol}, that \eqref{symop}  extends to a bounded operator with norm uniformly bounded on compact subsets of $U$.
 Since $ \inn{ \phi_1 , \eqref{symop}  \phi_2}$ is by assumption   for $\phi_1 , \phi_2 \in \mathcal{D}(-\Delta + H_f)$
analytic on $U$  (see e.g. \cite[Theorem XII.7]{ReedSimon.1978})    the analyticity of  the closure of \eqref{symop} now follows,
because  weak analyticity implies strong analyticity, cf.  Lemma \ref{thm:ana2}.
 \end{proof}

\begin{proof}[Proof of Theorem \ref{propverihypD}] First observe that  by Hypothesis \ref{GSanahyp00} there exists
a ground state.
First we show Hypothesis  \ref{hypresikvebt} for the set $S_j$.
  For  fixed $k \in \R^3 \setminus \{ 0 \}$ we consider  the dilation of  the operator $D_1(k,\lambda)$ for $\theta \in \R$
\begin{align*}
 &  D_1(k,\lambda;\theta)    :=  U(\theta) D_1(k,\lambda)  U(\theta)^{-1}  \\
&  \ =   \frac{  1 }{(2\pi)^{3/2}}  \sum_{l=1}^N \left\{    ( e^{-\theta} p_l +  A_\theta(x_l)  )  \cdot 2  \frac{ \varepsilon(k,\lambda)  }{\sqrt{2 |k|} } \overline{\kappa(k)} e^{  i k \cdot  e^\theta x_l}+     \mu   S_l \cdot \frac{ i k \wedge   \varepsilon(k,\lambda) }{\sqrt{2 |k|} }  \overline{\kappa(k)} e^{  i k  \cdot e^\theta x_j}\right\}   ,
\end{align*}
(note we interchanged the factors of the first term  on the right hand side, which is justified  by  $\epsilon(k,\lambda) \cdot k = 0$).
Now choose $\theta_0 \in (0,\min\{ \tilde{\theta}_0,\pi/4 \})$ such that the assertions of  Lemma ~\ref{lem:anaext},
Lemma~\ref{remeasyspecdef} for $x = E_{{\rm at},j}$, and
 Hypothesis~\ref{GSanahyp00} hold   for some  $g_0 > 0$ and   $\beta > 0$. By possibly making  $\theta_0 > 0$ smaller,  we can assume without loss that
\begin{equation} \label{estEvsG}
\theta_0 <      \frac{ \beta e^{-\pi/4}}{  E_{{\rm at}, j} - E_{\rm gs}  +  \delta_j/2  } =
   \frac{ \beta e^{-\pi/4}}{  \sup \{   |k|  : k \in S_j \}  }   .
\end{equation}
  It  follows  from \eqref{estEvsG} that  for $k \in S_j$
\begin{align} \label{anaest333}
| {\rm Re } (  i  k \cdot e^\theta x_l )  | = e^{{\rm Re} \theta} | \sin({\rm Im } \theta) |   |k| |x|  \leq e^{\pi/4}  \theta_0
 |k| |x|  < \beta |x|  , \quad l=1,...,N .
 \end{align}
 It follows   with  $X := (p^2 + H_{\rm f} + 1 )$ that for $k \in S_j$
 \begin{align}  & X^{-1/2}   D_{1} (k,\lambda; \theta) \psi_{{\rm gs}, \theta}  \label{stateana}  \\
& = \frac{  1 }{(2\pi)^{3/2}}  \sum_{l=1}^N 2 X^{-1/2}  ( e^{-\theta} p_l +  A_\theta(x_l)  ) \cdot   \frac{ \varepsilon(k,\lambda)  }{\sqrt{2 |k|} } \overline{\kappa(k)}  e^{  i k \cdot  e^\theta x_l}  e^{- \beta|x|} e^{\beta |x| }  \psi_{{\rm gs}, \theta} \nonumber  \\
&  +    X^{-1/2}  \sum_{l=1}^N   \mu   S_l \cdot \frac{ i k \wedge   \varepsilon(k,\lambda) }{\sqrt{2 |k|} }  \overline{\kappa(k)}  e^{  i k  \cdot e^\theta x_l} e^{- \beta|x|} e^{\beta |x| } \psi_{{\rm gs}, \theta} \nonumber
\end{align}  is an $\HH$-valued analytic function of $\theta$  on $ D_{\theta_0}$
as   a consequence of    Lemma  \ref{lem:anaext},     Remark   \ref{GSanahyp002} to Hypothesis    \ref{GSanahyp00} and
inequality  \eqref{anaest333}.

Now we observe  for any $z  \in \C$ with ${\rm Im} z > 0$ that  for $\theta \in \R$ and $k,k' \in \R^3 \setminus \{ 0 \}$
\begin{align}
& \inn{  D_{1}(k,\lambda) \psi_{\rm gs} , (H   - z  )^{-1}  D_{1} (k',\lambda')\psi_{\rm gs} } \label{eq:boundaryvalueres0}   \\
& = \inn{  D_{1}(k,\lambda;\theta) \psi_{{\rm gs}, \theta}  , (H(\theta)   - z )^{-1}  D_{1} (k',\lambda';\theta) \psi_{{\rm gs}, \theta} }   \nonumber \\
& = \inn{ X^{-1/2}  D_{1}(k,\lambda;\theta) \psi_{{\rm gs}, \theta}  ,  X^{1/2} (H(\theta)   - z )^{-1}  X^{1/2} X^{-1/2} D_{1} (k',\lambda';\theta) \psi_{{\rm gs}, \theta} }  \nonumber   ,
\end{align}
where we used the  unitarity of  dilation and a   trivial insertion of the identity $1 = X^{1/2} X^{-1/2}$.

Next we consider an analytic continuation of \eqref{eq:boundaryvalueres0}.  First,   for  fixed  spectral parameter
$z $ in the complex upper half
plane we extend $\theta$ into the complex  lower half plane, and then, for fixed $\theta$ in the lower half plane we extend the spectral parameter
$z$ from the complex upper half plane across the real axis into the complex lower half plane.
By Lemma \ref{remeasyspecdef}  we  can choose   $ y \in (0,\infty)$ sufficiently
large such that  $E_{{\rm at},j } + i y \in \rho(H(\theta) )$
 for all $g  \in (-g_0,g_0)$ and  $\theta \in D_{\theta_0}$ with ${\rm Im} \theta \geq  0$.
Thus it follows that for $z = z_0  :=  E_{{\rm at},j } + i y$ the right hand side of  \eqref{eq:boundaryvalueres0}
 is complex differentiable  in  points $\theta \in D_{\theta_0}$ with   ${\rm Im}\theta \geq 0$,
 by the analyticity of  \eqref{stateana}   and the analyticity of  $X^{1/2} (H(\theta)  - z_0  )^{-1}  X^{1/2}$, which in turn holds  by
 Lemma \ref{symop2}. Since  the left  hand side  of \eqref{eq:boundaryvalueres0} does not depend on $\theta$
 we conclude  by analytic continuation that  for $\vartheta \in (0,\theta_0)$  and $z = z_0$
 \begin{align}
& \inn{  D_{1}(k,\lambda) \psi_{\rm gs} , (H   - z  )^{-1}  D_{1} (k',\lambda')\psi_{\rm gs} }\label{eq:boundaryvalueres2}   \\
& = \inn{ X^{-1/2}  D_{1}(k,\lambda;i \vartheta) \psi_{{\rm gs}, i \vartheta}  ,  X^{1/2} (H(i \vartheta)   - z  )^{-1}  X^{1/2} X^{-1/2} D_{1} (k',\lambda'; i \vartheta) \psi_{{\rm gs}, i \vartheta} }     .   \nonumber
\end{align}
Now fix $\vartheta \in (0, \theta_0)$.  Since $H$ is self-adjoint  the left hand side of  \eqref{eq:boundaryvalueres2} is analytic in $z$ on the upper half complex plane.
By  Theorem  \ref{bfs99thm3.2}  there exists a $c_j > 0$  and a $g_1 \in ( 0, g_0)$ such that
\begin{align} \label{resass}  \mathcal{A}(E_j,\delta_j,c_j) \subset \rho(H(i \vartheta)) \end{align}  for $0 < |g|<g_1$.
Thus the right hand side of \eqref{eq:boundaryvalueres2} is an anlytic function of $z \in \mathcal{A}(E_j,\delta_j,c_j)$.
Since  \eqref{eq:boundaryvalueres2}  holds for $z= z_0$ and $z_0 \in \mathcal{A}(E_j,\delta_j,c_j)$
we conclude that the left hand side of \eqref{eq:boundaryvalueres2}  has an analytic continuation to  $z \in \mathcal{A}(E_j,\delta_j,c_j)$,
and that this analytic continuation satisfies for $0 < |g|  < g_1$   and all $z \in   \mathcal{A}(E_j,\delta_j,c_j)$ the identity
 \begin{align}
& \inn{  D_{1}(k,\lambda) \psi_{\rm gs} , (H    - z  )^{-1}  D_{1} (k',\lambda')\psi_{\rm gs} }\label{eq:boundaryvalueres22}   \\
& = \inn{ X^{-1/2}  D_{1}(k,\lambda;i \vartheta) \psi_{{\rm gs}, i \vartheta}  ,  X^{1/2} (H(i \vartheta)  - z )^{-1}  X^{1/2} X^{-1/2} D_{1} (k',\lambda'; i \vartheta) \psi_{{\rm gs}, i \vartheta} }     .   \nonumber
\end{align}
Since $k' \in S_j$ implies  $E_{\rm gs} + |k'| \in \mathcal{A}(E_j,\delta_j,c_j)$, we
 see from  \eqref{resass}  by inserting $z = E_{\rm gs} + |k'|$ into  \eqref{eq:boundaryvalueres22}  that   Hypothesis  \ref{hypresikvebt} holds.

Next we show regularity of \eqref{defofT}   in $|k|$ and $|k'|$. Let  $k \in S_j$ and $k' \in \R^3 \setminus \{ 0 \}$.
By what we showed  we can evaluate  \eqref{eq:boundaryvalueres22} at the point $z = E_{\rm gs}  +  \omega(k') \in  \mathcal{A}(E_j,\delta_j,c_j)$ and find
 \begin{align}
& \inn{  D_{1}(k,\lambda) \psi_{\rm gs} , (H - E_{\rm gs}  - \omega(k')  )^{-1}  D_{1} (k',\lambda')\psi_{\rm gs} }  \label{eq:boundaryvalueres222} \\
& = \inn{ X^{-1/2}  D_{1}(k,\lambda;i \vartheta) \psi_{{\rm gs}, i \vartheta}  ,  X^{1/2} (H(i \vartheta) - E_{\rm gs}  - \omega(k') )^{-1}  X^{1/2} X^{-1/2} D_{1} (k',\lambda'; i \vartheta) \psi_{{\rm gs}, i \vartheta} }     .   \nonumber
\end{align}
First observe that   the vector given  in \eqref{stateana} is  for any $\theta \in D_{\theta_0}$  arbitrary many  times  differentiable  in $|k|$ for $k \neq 0$
(as  a $\HH$--valued differentiable  function), since $k \mapsto   e^{  i k \cdot  e^{ i \vartheta}  x_j}  e^{- \beta|x|}$ is for $k \in S_j$  a $C^\infty$--map  into the bounded operators on $\HH_{\rm at}$, cf. \eqref{anaest333}.
Clearly,  $ k \mapsto \omega(k)$ is on $S_j$ a $C^\infty$--map  and the  resolvent is an analytic and hence $C^\infty$--function of its spectral parameter. We thus conclude   that
the right hand side of \eqref{eq:boundaryvalueres222}
is infinitely often differentiable as a function of    $|k'|$ for $k' \in S_j$.
Similarly (but simpler as it does not involve derivatives with respect to the spectral parameter of the resolvent), one  shows for $k' \in \R^3 \setminus \{ 0 \}$   the  $C^\infty$-differentiability  as a function of $|k|$ .
 The treatment of the second term in \eqref{defofT} is similar but  simpler.
Indeed,  the resolvent map is always evaluated   in the resolvent set, so no analytic dilation is necessary.
Finally, that   the last term in  \eqref{defofT}, i.e.,
\begin{align*}
& \inn{ \psi_{{\rm gs}}  , D_2(k,\lambda,k' , \lambda') \psi_{{\rm gs}} }  \\
& = \left\langle   \psi_{{\rm gs}}  ,  \frac{ 2 }{(2\pi)^{3}}  \sum_{l=1}^N  e^{  i ( k - k') \cdot x_l}  \overline{  \kappa(k) } \kappa(k')   \frac{  \varepsilon(k,\lambda) \cdot \varepsilon(k',\lambda') }{\sqrt{2|k|} \sqrt{2|k'|}  }  \psi_{{\rm gs}}  \right\rangle
\end{align*}
is infinitely differentiable as a function of $|k|$ and $|k'|$  for $k, k' \in \R^3 \setminus \{ 0 \}$  follows easily  by  dominated convergence and the exponential decay of  $\psi_{{\rm gs}}$. This shows  Proposition   \ref{propverihypD}.
\end{proof}

By collecting the results of this section and combining them with Theorem   \ref{main:qed}   we obtain
the following result, recalling the notation introduced in   \eqref{stucspec},   \eqref{defofZj},   and \eqref{defofdeltaj}.

\begin{theorem}  \label{main:qedconc}  Suppose $V= V_C$,   \eqref{modelasskappa} and   that Hypotheses
\ref{thm:fgs01thm4} and  \ref{GSanahyp00} hold.  Let $j \in \{1,..., M \}$. If  ${\rm Im} Z_j  > 0$,
 then there exists a $g_0> 0$
such that for   $0 < |g| < g_0$ the followoing holds for the set
$$
S_j  =  \{ k \in \R^3 \setminus \{ 0 \} :
 E_{{\rm at}, j} - E_{\rm gs} -  \delta_j/2  <  \omega(k)  <  E_{{\rm at}, j} - E_{\rm gs}  +  \delta_j/2 \}  .
$$
For  $f, h  \in C_c(S_j)^2$
we have
\begin{align}
 & \inn{ a^*_{\rm out}(f) \psi_{\rm gs} , a^*_{\rm in}(h) \psi_{\rm gs} }  - \inn{ f , h }  \label{eq:mainform2} \\
& = - 2 \pi  i \sum_{\lambda, \lambda'}\int_{\R^3}  \int_{\R^3} \overline{ f(k , \lambda)} \delta(\omega(k) - \omega(k') )   h(k',\lambda')  T(k,\lambda,k', \lambda') dk'  dk  .  \nonumber
\end{align}
\end{theorem}

\begin{remark}  {\rm Hypothesis  \ref{GSanahyp00}   has been verified in \cite{HaslerLejsek.2022} for  $V= V_C$,  \eqref{modelasskappa}, spinless ``electrons''  and small values of the coupling constant $g$ in case the atomic Hamiltonian  has a   nondegenerate ground state  (which is known to be  the
case for the  Hydrogen atom).   }
\end{remark}

\begin{remark}{\rm
We note that Hypothesis \ref{thm:fgs01thm4} has been shown in
 \cite[Theorem 4]{FrohlichGriesemerSchlein.2001} in the spinless case.
The result in  \cite{FrohlichGriesemerSchlein.2001} was  proven for  models,
for which  $V = V_C$ holds, 
but with the additional assumption  that $\tilde{\kappa}$ vanishes outside a compact set.
There is no obvious reason  why  the result  in  \cite{FrohlichGriesemerSchlein.2001} as well as its  proof
do not carry over to situations where $\tilde{\kappa}$ is merely exponentially decaying  (as was noted by one of the authors).
 In particular,  given that in \cite{FrohlichGriesemerSchlein.2002} an analogous result was shown
 for a the  related  Nelson model, where only exponential falloff   for  the coupling function   was  assumed.

}
\end{remark}

\begin{remark} \label{hasherbhub} {\rm
For the hydrogen atom $V = V_C$ with  $Z=N=1$   and  \eqref{modelasskappa},
with  the physically natural assumption $\lim_{r \downarrow 0} \tilde{\kappa}(r)  = 1$,
we believe that for any  $j \geq 2$ and  $\alpha >0$  sufficiently small   Fermi's golden rule condition
 ${\rm Im} Z_j > 0$   holds.
 Explicitly, this should follow by means of the   scaling
$ U_{\rm el}(\xi) (-\Delta - \alpha |x|^{-1} ) U_{\rm el}^{-1}(\xi ) =  \alpha^2 (-\Delta - |x|^{-1} )$, with $\xi = - \ln \alpha$,  from an  expansion  in $\alpha$ around 0
and the explicit calculation  in  \cite[Theorem B.1]{HaslerHerbstHuber.2008}. }
\end{remark}

\section*{Acknowledgements}

The author wants to thank Christian  Lejsek for valuable discussions.

\appendix

\section{Estimates for  creation and annihilation operators}

The following lemma gives elementary estimates of the creation operators in terms of the free field operator.

\begin{lemma} \label{lem:elemest}  For each  $n \in \N$ there  exists a finite constant $C_n$ such that for all  $h_i \in L^2_\omega(\R^3 \times \Z_2)$, $i=1,...,n$, the inequality
\begin{align*}
\| a^\#(h_1) \cdots  a^\#(h_n) (H_{\rm f} + 1)^{-n/2} \| \leq  C_n \prod_{i=1}^n \| h_i \|_\omega
\end{align*}
holds.
\end{lemma}
A proof of Lemma  \ref{lem:elemest}  can be found  for example in \cite{FrohlichGriesemerSchlein.2001}. 

\begin{lemma}  \label{FGS.Lemma18} Suppose that   \eqref{eq:potas} and \eqref{eq:assonfield} hold.  Then for all $\epsilon > 0$ there exists a constant $D_\epsilon$ such that the following
holds  in the form sense. We have
\begin{equation} \label{ineqopbounded}
V_- \leq \epsilon H + D_\epsilon .
\end{equation}
If   $s=1/2$, then  in addition
\begin{equation} \label{ineqopbounded223}
 |  B_l(x_j)  |   \leq \epsilon H + D_\epsilon  ,
\end{equation}
for $l=1,2,3$ and $j=1,...,N$.
\end{lemma}
\begin{proof} We note that the case $s=0$ has been shown in the proof of  Lemma 18  in \cite{FrohlichGriesemerSchlein.2001}. We will only show the case $s=1/2$, and thus assume  $s=1/2$.
For  this   we  use that for any operators  $X$ and $Y$ in a Hilbert space  the
following inequality holds in  sense of forms
\begin{equation} \label{opineq}
X^* Y + Y^* X \leq X^2 + Y^2   ,
\end{equation}
which follows from a simple application of  Cauchy's inequality.
For  $\delta > 0$, we see from \eqref{opineq}, the trivial   operator
inequality $Z^* Z \leq \| Z \|^2$,  and   Lemma \ref{lem:elemest}
that
\begin{align}
|B_l(x_j)|   & \leq   \frac{1}{2} \delta B_l(x)^2 + \frac{1}{2} \delta^{-1} \nonumber  \\
&
 \leq \frac{1}{2} \delta  (H_{\rm f} + 1 )^{1/2}  (H_{\rm f} + 1)^{-1/2}  B_l(x)^2(H_{\rm f}+1)^{-1/2}   (H_{\rm f} + 1 )^{1/2} +  \frac{1}{2} \delta^{-1} \nonumber \\
&
 \leq \frac{1}{2} \delta  (H_{\rm f} + 1 )^{1/2} \| B_l(x) (H_{\rm f}+1)^{-1/2} \|^2    (H_{\rm f} + 1 )^{1/2} +  \frac{1}{2} \delta^{-1} \nonumber \\
&
 \leq \frac{1}{2} \delta C  (H_{\rm f} + 1 )   +  \frac{1}{2} \delta^{-1} , \nonumber
\end{align}
for some constant $C$ depending on the norm $\| J_{x,l} \|_\omega$.
Since $\delta > 0$ is arbitrary, we see that  for any $\epsilon > 0$ there is  a constant $E_\epsilon$ such that
\begin{align}
|B_l(x_j)|
&  \leq \epsilon H_{\rm f} + E_\epsilon  .  \label{eq:boundonabsB}
\end{align}
Since  $A(x)^2 \leq  C(H_{\rm f} + 1)$, by  Lemma \ref{lem:elemest}, and
\begin{align*}
p_j^2 & = (p_j-A(x_j) + A(x_j)  )^2 \\
&  \leq  (p_j-A(x_j))^2 + (p_j-A(x_j) ) \cdot A(x_j) + A(x_j)  \cdot (p_j-A(x_j))   + A(x_j)^2  \\
& \leq 2 ((p_j + A(x_j))^2 + 2 A(x_j)^2 ,
\end{align*}
by  \eqref{opineq} again,
it follows  using non-negativity that
$$
\sum_{j=1}^N p_j^2  + H_{\rm f} \leq a ( H + V_- - \sum_{j=1}^N \mu S_j \cdot B(x_j)   ) + b
$$
for some $a,b > 0$. Combining this with  \eqref{eq:potas} and     \eqref{eq:boundonabsB}
it follows that $\sum_{j=1}^N p_j^2 + H_{\rm f}$ is form bounded with respect to $H$.
This  with     \eqref{eq:potas} shows   \eqref{ineqopbounded}, and with \eqref{eq:boundonabsB} it shows  \eqref{ineqopbounded223}.
\end{proof}

For a function  $f: Y^n \to Z$ and  $\sigma\in  \mathfrak{S}_n$  we define \begin{equation} \label{eq:permfun}
(\sigma f)(y_1,...,y_n)  = f(y_{\sigma(1)},...,y_{\sigma(n)}) , \quad y_j \in Y , \ j=1,...,n .   \end{equation}
Fock spaces   over    $L^2$-spaces  can be canonically  identified with direct sums of subspaces of
 $L^2$-spaces over Cartesian products.
 In the situation of non-relativistic qed we obtain the
following identification, see for example  \cite[Theorem 11.10 (a)]{ReedSimon.1972},
\begin{equation} \label{eq:isomeas}
\FF_n(\hh) \cong L_{\rm s}^2((\R^3 \times \Z_2)^n) :=  \{ \psi \in L^2((\R^3 \times \Z_2)^n) : \sigma \psi  = \psi , \ \forall \sigma \in \mathfrak{S}_n\} .
\end{equation}
Thus we can identify  $\psi \in \FF(\hh)$ with a sequence $\psi = (\psi_n )_{n \in \N_0}$, where    $\psi_n \in L_{\rm s}^2((\R^3 \times \Z_2)^n)$. We define   for $k \in \R^3$ and $\lambda \in \Z_2$
the  action of the  formal annihilation operator $a(k,\lambda)$  by
\begin{equation} \label{eq:defofpsi}
(a(k,\lambda) \psi)_n(k_1,\lambda_1,...,k_n,\lambda_n) := (n+1)^{1/2} \psi_{n+1}(k,\lambda,k_1,\lambda,\cdots,k_n,\lambda_n) ,  n \in \N_0 ,
\end{equation}
which is understood as  an identity of measurable functions (note that \eqref{eq:defofpsi} does not necessarily need to define  an element of   Fock space).

\begin{lemma}\label{invdom}  Let $\psi \in \mathcal{D}(H^2)$. Then $a^*(f) \psi \in \mathcal{D}(H)$ for all $f \in \hh_\omega$ with $\omega f \in \hh_\omega$.
\end{lemma}
\begin{proof}
For $\psi \in \mathcal{D}(H^2)$ a straight forward calculation shows that  for all $\varphi \in \mathcal{C}$
\begin{align*}
\langle  H \varphi,  a^*(f) \psi \rangle &  = \langle a(f)  H \varphi ,  \psi \rangle \\
& = \langle [a(f), H ] \varphi, \psi \rangle + \langle H a(f) \varphi, \psi \rangle  \\
& = \langle  C \varphi , \psi \rangle +  \langle  \varphi, a(f)  H  \psi \rangle  ,
\end{align*}
where we  introduced the operator
$$
C =  \sum_{j=1}^N \sum_{l=1}^3  \frac{1}{\sqrt{2}} \left\{  \langle  f , G_{x_j, l } \rangle  ( p_{j,l}  +  A_l(x_j) ) +( p_{j,l}  +  A_l(x_j) )  \langle  f , G_{x_j, l } \rangle   + \mu S_{j,l} \langle f ,  J_{x_j,l} \rangle  \right\}   + a(\omega f )
$$
on $\mathcal{C}$.  From Lemma \ref{lem:elemest} we see that the domain of the adjoint of $C$ contains the natural
domain of $(-\Delta + H_{\rm f})^{1/2}$ and so contains the   the domain of $\mathcal{D}(H)$, by  \eqref{domprob}.  It follows that
\begin{align*}
\langle  H \varphi,  a^*(f) \psi \rangle &  = \langle\varphi ,   C^* \psi \rangle +  \langle  \varphi, a(f)  H  \psi \rangle .
\end{align*}
Since $H$ is by assumption essentially self-adjoint on $\mathcal{C}$, it follows that $a^*(f) \psi  \in \mathcal{D}(H)$.
\end{proof}

\section{A lemma about Abelian limits}

We will need the following proposition about Abelian limits. Together with its proof it can be found in
 \cite[XI.6, Lemma 5]{ReedSimon.1979}. For the convenience of the reader, we provide  the   proof given there.

\begin{proposition}\label{lem:limintident}  Let $f$ be a bounded measurable function on $[0,\infty)$  and suppose that $$\lim_{t \to \infty} \int_0^t f(x) dx = a .$$ Then $\lim_{\epsilon \downarrow 0} \int_0^\infty e^{-\epsilon s} f(s) ds = a $.
\end{proposition}
\begin{proof}
Let $g(t) = \int_0^t f(s) ds$ and $q(\epsilon) = \int_0^\infty e^{- \epsilon s} f(s) ds$.
Then $g'(t)=f(t)$, a.e., so integration by parts,   change of variables, and  $\int_0^\infty e^{-s} ds = 1$  implies  that
\begin{align*}
q(\epsilon) & = \int_0^\infty \epsilon e^{-\epsilon s} g(s) ds  = \int_0^\infty  e^{- s} g(\epsilon^{-1} s) ds \\
& = \int_0^\infty  e^{- s} ( g(\epsilon^{-1} s) - a ) ds + a   \\
&  \to 0 + a \quad (\epsilon \downarrow 0 ) ,
\end{align*}
where in the last line we used  that    $g$ is bounded on $\R_+$,   $\lim_{t \to \infty} g(t)  =  a$, and dominated convergence
\end{proof}

\section{Results from Functional Analysis}

We will use the following result which is a version of the well known fact   that weak analyticity implies strong analyticity.

\begin{theorem} \label{thm:ana2}  Let $X$ be a Banach space and $L$ a linear subspace of  $X^*$
such that  $\| x \| = \sup_{l \in L , \| l \| \leq 1} | l(x) |$ for all $x \in X$. Let $D$ be open and $f : D \to X$
a map  such that for all $l \in L$ the  composition $l \circ f : D \to \C$ is analytic and
$\sup_{z \in K} \sup_{l \in L , \| l \| \leq 1 } |l(f(z))| < \infty$ on compact subsets $K$ of $D$. Then $f$ is
strongly analytic.
\end{theorem}
The above theorem follows as a  consequence of
\cite[III. \S 1 Theorem 1.37 and Remark 1.38]{Kato.1980}
 For convenience  of the reader  we provide  a proof.
\begin{proof}
Let $a \in $ and suppose that $\Gamma$ is a circle in $D$ containing $a$,
whose interior is contained in $D$. If $l \in L$, then by assumption $l \circ f$ is analytic and so by Cauchy's formula
\begin{align*}
& l\left( \frac{ f(a+h) - f(a)}{h} \right)  - \frac{d}{dz} l(f(a)) \\
& =  \frac{1}{2\pi i } \oint_\Gamma \left[ \frac{1}{h} \left( \frac{1}{z-(a+h)} - \frac{1}{z-a} \right) - \frac{1}{(z-a)^2} \right] l(f(z)) dz .
\end{align*}
Thus using the triangle inequality we find with $C_\Gamma :=  \sup_{z \in \Gamma}  \sup_{l \in L , \| l \| \leq 1 } |l(f(z))|$ that
\begin{align*}
& \left\|  \frac{ f(a+h) - f(a)}{h}   - \frac{ f(a+h') - f(a)}{h'}  \right\|  \\
& \leq   \frac{1}{2\pi  } C_\Gamma  \int_\Gamma  \left(  \left| \frac{1}{(z-(a+h))(z-a)}  - \frac{1}{(z-a)^2} \right| +  \left| \frac{1}{(z-(a+h'))(z-a)}  - \frac{1}{(z-a)^2} \right|  \right) dS(z) ,
\end{align*}
where $dS$ denotes the surface measure of $\Gamma$ as a one dimensional real submanifold of $\C \cong \R^2$.
Now the right hand side tends to zero as $h$ and $h'$ tend to zero. It follows by completeness that  $\frac{ f(a+h) - f(a)}{h} $ converges in $X$, proving that $f$ is strongly analytic.
\end{proof}

\begin{lemma} \label{intpol}  Let $A$ be a self-adjoint non-negative  operator. Let $B$ be a bounded  operator with $\ran B \subset \mathcal{D}(A)$ such that for some constant $C$ we have
$$
 \| B  A \|  \leq C  , \quad  \| A  B  \|  \leq C .
$$
Then $ A^{1/2} B  A^{1/2} : \mathcal{D}(A) \to \HH$ is a  bounded operator with norm bounded by $C$.
\end{lemma}
\begin{proof}
We use interpolation.  Consider first  a regularization of $A$  as follows $A_n = \frac{n A}{n+ A}$, $n \in \N$. Then, clearly
\begin{align*}
 \| B A_n \| \leq  \| B  A \| \| (1 + A/n)^{-1} \| \leq  \| B  A \| \leq C \\
 \| A_n B   \| \leq   \| (1 + A/n)^{-1} \| \|  A  B  \| \leq  \| A B  \|  \leq C .
\end{align*}
It follows from interpolation, see for example \cite[Proposition 1 in  Appendix to IX.4]{ReedSimon.1975} which  also holds for bounded operators,
that for all $n \in \N$
 $$\| A_n^{1/2}  B A_n^{1/2}  \| \leq C .  $$
Thus
$$  \inn{ f , A^{1/2}  B  A^{1/2} g } = \lim_{n \to \infty}  \inn{ f , A_n^{1/2} B A_n^{1/2} g } $$
for all $f, g \in D(A)$. We conclude that $A^{1/2} B  A^{1/2}$ is bounded by $C$.
\end{proof}

%
%

\bibliographystyle{abbrv}
\bibliography{lit}

\end{document}